\documentclass[journal,twoside,web]{ieeecolor}
\usepackage{generic}
\usepackage{cite}
\usepackage{amsmath,amssymb,amsfonts}
\usepackage{algorithmic}
\usepackage{graphicx}
\usepackage{algorithm,algorithmic}
\usepackage{hyperref}
\usepackage{comment}
\usepackage{textcomp}
\usepackage{xcolor}

\usepackage{amsthm}
\usepackage{arydshln}
\usepackage{empheq}
\usepackage{booktabs}   

\newtheorem{theorem}{Theorem}
\newtheorem{assumption}{Assumption}
\newtheorem{corollary}{Corollary}[theorem]
\newtheorem{lemma}[theorem]{Lemma}

\newtheorem{definition}{Definition}
\newtheorem{remark}{Remark}

\newcommand*{\tm}{\textcolor{teal}}

\DeclareMathOperator{\rank}{rank}  
\DeclareMathOperator{\bdiag}{blkdiag}  
\DeclareMathOperator{\vecop}{vec}
\DeclareMathOperator{\col}{col}
\newcommand{\RR}{\mathbb{R}}
\newcommand{\CC}{\mathbb{C}}
\newcommand{\ZZ}{\mathbb{Z}}
\newcommand{\ZZn}{\mathbb{Z}_{\geq 0}}
\newcommand{\ZZp}{\mathbb{Z}_{> 0}}
\newcommand{\ie}{\textit{i.e.}}
\newcommand{\eg}{\textit{e.g.}}
\newcommand{\image}{\text{Im}}

\newcommand{\Pir}{\Pi_r}     
\newcommand{\Gar}{\Gamma_r}  

\def\BibTeX{{\rm B\kern-.05em{\sc i\kern-.025em b}\kern-.08em
    T\kern-.1667em\lower.7ex\hbox{E}\kern-.125emX}}
\newcommand{\GS}[1]{\textcolor{orange}{#1}}

\begin{document}
\title{One Equation to Rule Them All---Part II: Direct Data-Driven Reduction and Regulation}
\author{Junyu Mao, \IEEEmembership{Student Member, IEEE}, Emyr Williams, \IEEEmembership{Student Member, IEEE}, 
Thulasi Mylvaganam, \IEEEmembership{Senior Member, IEEE}, and Giordano Scarciotti, \IEEEmembership{Senior Member, IEEE}
\thanks{J. Mao, E. Williams, and G. Scarciotti are with the Department of Electrical and Electronic Engineering, Imperial College London, SW7 2AZ London, U.K. T. Mylvaganam and E. Williams are with the Department of Aeronautics, Imperial College London, SW7 2AZ London, U.K.
(e-mail:  junyu.mao18@ic.ac.uk; emyr.williams18@ic.ac.uk; t.mylvaganam@ic.ac.uk; g.scarciotti@ic.ac.uk). }}

\maketitle

\begin{abstract}
The Sylvester equation underpins a wide spectrum of control synthesis and systems analysis tools associated with cascade interconnections. In the preceding Part I \cite{mao2025partOne} of this article, it was shown that such an equation can be reformulated using data, enabling the production of a collection of data-driven stabilisation procedures. In this second part of the article, we continue to develop the framework established in Part I to solve two important control-theoretic problems: model order reduction and output regulation. For the model order reduction problem we provide a solution from input-state measurements, from input-output measurements, and we study the effect of the noise. For the output regulation problem, we provide data-driven solutions for the static and dynamic feedback problem. The proposed designs are illustrated by means of examples. 
\end{abstract}

\begin{IEEEkeywords}
Control design, data-driven control, learning systems, linear matrix inequalities, robust control, output regulation, model order reduction
\end{IEEEkeywords}

\section{Introduction}
\label{sec:introduction}
Modern engineering systems consist of a massive number of interconnected subsystems, \eg, power systems \cite{farhangi2009path}, robotic systems \cite{rubenstein2014programmable}, automotive vehicles \cite{koscher2010experimental}, and smart cities \cite{zanella2014internet}. To address classes of data-driven design problems regarding these systems, in Part I \cite{mao2025partOne} of this article, an interconnection-based framework has been proposed. The unifying perspective provided by this framework is to decompose the complex tasks into interconnections of two or multiple dynamical subsystems, based on which a systematic treatment of these tasks using a shared methodology has been demonstrated.

In Part II of this article, we continue to develop
this ``interconnection'' perspective formulating two additional (data-driven) control problems---model order reduction and output regulation---and show how these can be tackled by this methodological framework. 
In the following, we review the literature related to these two problems. While a high volume of research exists for both problems, we restrict the reported literature to works that are either classic, related to the Sylvester equation, or data-driven, for reasons of space.

{\small \sf \color{nblue} \textit{Model Order Reduction:}} Finding a low-order description of a high-order system while preserving its dominant behaviour and key properties is a well-established pathway for mitigating the ``curse of dimensionality'' encountered in a growing number of engineering applications. To this end, model order reduction techniques have been extensively studied, initially from a model-based perspective, with notable examples including the Hankel-norm approximation \cite{adamjan1971analytic}, balanced truncation \cite{moore1981principal}, and Krylov methods (also referred to as \textit{moment matching}) \cite{kimura1986positive, byrnes1995complete, antoulas1990solution, astolfi2010model}. We refer the reader to \cite{antoulas2005approximation} for an extensive overview of model-based techniques and to \cite{scarciotti2024interconnection} for a detailed survey of moment matching approaches which are based on the solutions of Sylvester equations. Model order reduction has also been approached extensively from a data-driven perspective; see, \eg, proper orthogonal decomposition \cite{kunisch1999control}, dynamic mode decomposition \cite{schmid2010dynamic}, data-driven balanced truncation \cite{willcox2002balanced}, and the Loewner framework (mainly based on frequency-domain data) \cite{mayo2007framework}. A further class of approaches is provided by time-domain moment matching \cite{scarciotti2017data, carnevale2021data, mao2024data, burohman2024data}, to which our contribution belongs. However, with respect to the literature, our method provides a noise analysis and/or differ in the experimental settings.

{\small \sf \color{nblue} \textit{Output Regulation:}} Enforcing the system to track a reference signal while rejecting disturbances is a core objective in a majority of control applications. A classical formulation of this problem, called \textit{output regulation problem} has been thoroughly investigated for linear systems in early works, \eg, \cite{francis1975internal, francis1976internal, davison1976robust}, where the well-known \textit{internal model principle} is shown to solve the regulation problem robustly. The core of the solution of this problem when using the static feedback relies on the \textit{regulator equations}, where a Sylvester equation plays a fundamental role, see, \eg{}, \cite{francis1977linear, isidori1990Output}. For a comprehensive collection of the classical model-based results, we direct the reader to \cite{huang2004nonlinear}.
Recently, multiple data-driven results have been established along alternative research lines, \eg{}, adaptive dynamic programming (\ie, value iteration and policy iteration) \cite{gao2016adaptive, lin2024data}, online gain tuning \cite{chen2023data}, and direct data-driven methods \cite{trentelman2021informativity, zhai2024data, hu2025data, liu2025data}. Our results for output regulation lie in the latter group, but we approach the problem from a novel interconnection angle that, to the best of the authors' knowledge, none of the literature provides.

This manuscript constitutes Part II of a two‑part article; in Part I \cite{mao2025partOne} we introduced a framework for directly computing the solution of the Sylvester equation from data samples and we studied its application to the problem of data-driven cascade stabilisation. The central idea of the framework is to decompose complex tasks into interconnections of two (or multiple) dynamical systems, offering a unifying perspective that departs from existing literature. For instance, while specific problems such as data-driven output regulation or data-driven model order reduction have been studied, in some instances, in the literature (as mentioned above), our approach is embedded in a broader theoretical context that enables a systematic treatment of multiple tasks using a shared methodology. A second key contribution developed in Part I lies in a comprehensive noise analysis, which is made possible precisely due to the interconnection-based formulation. Our framework enables the propagation and quantification of noise effects through multiple interconnected components, allowing us to derive explicit performance bounds on the final closed-loop systems directly from noise characteristics in the input data. This capability distinguishes our method from traditional two-step approaches, such as system identification followed by controller design, where such end-to-end noise guarantees are generally unavailable. The ability to track noise through the system interconnections, particularly in tasks like data-driven model order reduction and output regulation, represents a significant advantage and provides a further justification for the choice of a direct data-driven approach over two-step approaches. These results underpin the further developments established in this article. 

{\small \sf \color{nblue} \textbf{Contributions.}} The main contributions of this manuscript are summarised as follows.
\begin{itemize}
    \item[(I)] We leverage the framework introduced in Part I to determine the ``moments'' of the system via input-state and input-output data, respectively. The computation of moments underpins the construction of reduced-order models that interpolate the transfer function of the original system at prescribed complex interpolation points. 
    \item[(II)] We also provide a bound on the mismatch between the reduced-order model and the original system in the presence of measurement and process noise. 
    \item[(III)]  Building on the cascade stabilisation results established in Section IV of Part I, in Section~\ref{sec:outputRegulation} of this part we provide a solution of the \textit{dynamic} state-feedback output regulation problem in a data-driven setting.
    \item[(IV)] We show that our solution can be trivially extended to solving the dynamic \textit{error}-feedback problem, as formulated in Algorithm~\ref{alg:dataDrivenRobustRegulation} and demonstrated in Section~\ref{sec:exampleRegulation}.
    \item[(V)] We also show that building on the Sylvester equation results established in Section III of Part I, in Appendix~\ref{sec:regulationFullInfo} we provide a solution to the \textit{static} state-feedback output regulation problem, in both the cases in which the matrix of the exosystem is assumed to be known or unknown.
\end{itemize}

{\small \sf \color{nblue} \textbf{Organisation.}} Section~\ref{sec:preliminaries} recalls the interconnection-based framework and some instrumental results established in Part I. Then, these are exploited to develop direct data-driven approaches to the problems of model order reduction (Section~\ref{sec:modelOrderReduction}) and output regulation (Section~\ref{sec:outputRegulation}), with each section concluded with an illustrative example. Section~\ref{sec:conclusion} concludes the article by highlighting several future research directions. In Appendix~\ref{sec:extendedConfig}, we present an auxiliary result on an extended setting of the framework, while in Appendix~\ref{sec:regulationFullInfo} we solve the \textit{static} feedback output regulation problem in our framework, and we notice that the result coincides with what has been found in \cite{trentelman2021informativity}.

{\small \sf \color{nblue} \textbf{Notation.}} We use standard notation. $\RR$, $\CC$ and $\ZZ$ denote the sets of real numbers, complex numbers and integer numbers. $\ZZn$ ($\ZZp$) denotes the set of nonnegative (positive) integers. The symbols $I$ and $\textbf{0}$ denote the identity matrix and the zero matrix, respectively, whose dimensions can be inferred from the context. $A^\top$ and  $\rank (A)$ indicate the transpose and the rank of any matrix $A$, respectively. The set of eigenvalues (singular values) of a matrix $A$ is denoted by $\lambda(A)$ ($\sigma(A)$). 
Given a matrix $A$ of full row rank, $A^\dagger$ represents the right inverse such that $AA^\dagger = I$. 
Given a symmetric matrix $A \in \mathbb{R}^{n \times n}$, $A \succ \textbf{0}$ ($A \prec \textbf{0}$) denotes that $A$ is positive- (negative-) definite. Given a matrix $A \in \mathbb{R}^{n \times m}$, the operator $\vecop(A)$ indicates the vectorization of $A$, which is the $nm \times 1$ vector obtained by stacking the columns of the matrix $A$ one on top of the other. $\|A\|_2$, $\|A\|_F$ and $\|A\|_\infty$ denote the spectral, Frobenius and infinity norms of matrix $A$, respectively. $\image(A)$ denotes the image of any matrix $A$. Given matrices  $X_1, \cdots, X_n$ (with the same number of columns), $\text{col}(X_1, \cdots, X_n)$ denotes their vertical concatenation. Given a vector $x$, the symbol $\|x\|_2$ ($\|x\|_F$) denotes its Euclidean (Frobenius) norm. 
The symbol $\otimes$ indicates the Kronecker product. The symbol $\iota$ denotes the imaginary unit. \\
We use capital versions of lower case letters to indicate the corresponding data matrices. For example, given a signal $x: \ZZn \to \RR^{n}$ and a positive integer $T \in \ZZp$, we define 
$$
\begin{array}{rl}
X_{-} &\!\!\!\!\!:=  \begin{bmatrix} 
    x(0) & x(1) & \cdots & x(T-1)
\end{bmatrix},\\[2mm]
X_{+} &\!\!\!\!\!:=  \begin{bmatrix} 
    x(1) & x(2) & \cdots & x(T)
\end{bmatrix}.
\end{array}
$$

\section{Preliminaries} \label{sec:preliminaries}

In this section, we recall, from Part I \cite{mao2025partOne}, the data-driven framework to solve the Sylvester equation directly from (noise-free or noisy) system samples, as well as the methodology developed therein to solve the data-driven cascade stabilisation problem. 

\subsection{Cascade Interconnection and Sylvester Equation} \label{sec:PartI-framework}
Consider the cascade interconnection of two discrete-time linear time-invariant systems described by
\begin{equation} \label{Eq:SigmaSigma1}
\mathbf{\Sigma_1}:\,\,\begin{cases}\,\begin{aligned}
x_1(k+1) &= \mathbf{A_1} x_1(k) + \mathbf{B_1} u_1(k) \\
y_1(k) &= \mathbf{C_1} x_1(k) 
\end{aligned}\end{cases} 
\end{equation}
and
\begin{equation}
\label{Eq:SigmaSigma2}
\mathbf{\Sigma_2}:\,\,\begin{cases}\,\begin{aligned}
x_2(k+1) &= \mathbf{A_2} x_2(k) + \mathbf{B_2} u_2(k) \\
y_2(k) &= \mathbf{C_2} x_2(k) 
\end{aligned}\end{cases} 
\end{equation}
with $x_1(k) \in \RR^{n_1}$, $u_1(k) \in \RR^{m_1}$, $y_1(k) \in \RR^{p_1}$,
$x_2(k) \in \RR^{n_2}$, $u_2(k) \in \RR^{m_2}$, $y_2(k) \in \RR^{p_2}$ and $\mathbf{A}_1, \mathbf{B}_1, \mathbf{C}_1, \mathbf{A}_2, \mathbf{B}_2$, and $\mathbf{C}_2$ of proper dimensions, in which $\mathbf{\Sigma_1}$ drives $\mathbf{\Sigma_2}$ via $u_2(k) = y_1(k)$. Throughout this article, we use the tuples $\mathbf{\Sigma_1} := (\mathbf{A_1}, \mathbf{B_1}, \mathbf{C_1})$ and $\mathbf{\Sigma_2} := (\mathbf{A_2}, \mathbf{B_2}, \mathbf{C_2})$ to represent subsystems \eqref{Eq:SigmaSigma1} and \eqref{Eq:SigmaSigma2}, and $\mathbf{\Sigma_1} \rightarrow \mathbf{\Sigma_2}$ to represent the cascade where $\mathbf{\Sigma_1}$ drives $\mathbf{\Sigma_2}$. 

It is well known, see \eg \cite{simpson-porco-arxiv}, that the so-called Sylvester equation
\begin{equation} \label{eq:SylEqGeneric}
    \mathbf{A_2} \mathbf{\Theta} -  \mathbf{\Theta} \mathbf{A_1} = - \mathbf{B_2} \mathbf{C_1},
\end{equation}
plays an important role in characterising the dynamical relation between $\mathbf{\Sigma}_1$ and $\mathbf{\Sigma}_2$. In particular, the solution $\mathbf{\Theta} \in \mathbb{R}^{n_2\times n_1}$ characterises the \textit{invariant subspace}
\begin{equation} \label{eq:invariant-subspace}
    \mathcal{M} =\{(x_1,x_2) \in \RR^{n_1 + n_2}: x_2 = \mathbf{\Theta} x_1\}, 
\end{equation}
of the cascade system $\mathbf{\Sigma_1} \rightarrow \mathbf{\Sigma_2}$. Thus, one can use $\mathbf{\Theta}$ to define a
coordinate transformation for $\mathbf{\Sigma}_2$, namely $\zeta := x_2 - \mathbf{\Theta} x_1$, under which the dynamics of $\mathbf{\Sigma}_2$ are transformed into 
\begin{equation} \label{eq:zeta-dynamics}
    \zeta(k+1) = \mathbf{A_2} \zeta(k) - \mathbf{\Theta} \mathbf{B_1} u_1(k).
\end{equation}
At this point the literature studies two different cases, depending on whether $\mathbf{\Sigma_1}$ is considered the ``main system'' (\eg, a system to be controlled) while $\mathbf{\Sigma_2}$ is considered an ``auxiliary system'' (\eg, an internal model), or \textit{vice versa}. In the first case, which appears in problems such as cascade stabilisation, dynamic output regulation, and the so-called swapped moment matching method, $\mathbf{\Sigma_1}$ represents the (or a) ``main system'' with $u_1 \not \equiv 0$ \cite{simpson-porco-arxiv}. In the second case, which appears for instance in static output regulation and direct moment matching, $\mathbf{\Sigma_1}$ represents an ``auxiliary system'', typically an autonomous signal generator, and thus $u_1 \equiv 0$ \cite{simpson-porco-arxiv}. In Part I, we covered the cascade stabilisation problem \cite{mao2025partOne}. In this Part II, we consider the direct moment matching problem in Section~\ref{sec:modelOrderReduction} and the dynamic output regulation problem in Section~\ref{sec:outputRegulation}. Furthermore, we consider the static output regulation problem in Appendix~\ref{sec:regulationFullInfo}.

\subsection{Solving the Sylvester Equation from Data} 
In Part I \cite{mao2025partOne}, a configuration of practical interest has been extensively studied under both the noise-free and noisy scenarios, namely when $\mathbf{\Sigma_1}$ is known while $\mathbf{\Sigma_2}$ is unknown. The knowledge of $\mathbf{\Sigma_1}$ is without loss of generality as it corresponds to cases in which either (i) the proposed data-driven procedure is sequential, or (ii) $\mathbf{\Sigma_1}$ is designed by the user rather than being a physical system. In this article, the model order reduction problem corresponds to case (ii) while the dynamic output regulation problem corresponds to case (i). In the following, we recall the result established in Section III of Part I to solve the Sylvester equation from the data of $\mathbf{\Sigma_2}$. 

\begin{lemma} (see \cite[Lemma 1]{mao2025partOne}) \label{lemma:computePI}
Consider the cascade $\mathbf{\Sigma_1} \rightarrow \mathbf{\Sigma_2}$, where $\mathbf{\Sigma_1}$ is known and $\mathbf{\Sigma_2}$ is unknown, and the associated Sylvester equation \eqref{eq:SylEqGeneric}.
Suppose that the available data matrices are such that
\begin{equation} \label{eq:rankXU}
    \rank \left(\begin{bmatrix}
        X_{2, -} \\ U_{2, -}
    \end{bmatrix}\right)
     = n_2 + m_2,
\end{equation} 
and that $\lambda(\mathbf{A_1}) \cap \lambda(\mathbf{A_2}) = \emptyset$. Then any matrix $G_\Theta \in \RR^{T \times n_1}$ that satisfies
\begin{subequations} \label{eq:dataRepPi} 
  \begin{empheq}[left=\empheqlbrace]{align}
    X_{2, +}  G_\Theta &= X_{2, -} G_\Theta \mathbf{A_1}   \label{eq:dataRepPi1} \\
    U_{2, -}  G_\Theta &= \mathbf{C_1}    \label{eq:dataRepPi2}
  \end{empheq}
\end{subequations}
is such that
\begin{equation} \label{eq:parameterizedPi}
    \mathbf{\Theta} := X_{2, -} G_\Theta
\end{equation}
is the solution of \eqref{eq:SylEqGeneric}.
Conversely, the solution of \eqref{eq:SylEqGeneric} can be written as in \eqref{eq:parameterizedPi}, with $G_\Theta$ solution of \eqref{eq:dataRepPi}. \end{lemma}

The linear matrix equalities (LMEs) \eqref{eq:dataRepPi} serve as a data-dependent reformulation of the Sylvester equation \eqref{eq:SylEqGeneric}. Therefore, via solving a linear programme (LP) searching for a feasible solution to \eqref{eq:dataRepPi}, the solution $\mathbf{\Theta}$ of \eqref{eq:SylEqGeneric} can be obtained.

Now consider the following noisy scenario. First, suppose that the (unknown) system $\mathbf{\Sigma_2}$ is corrupted as
\begin{equation}  \label{eq:sysNoises}
     \quad x_2(k+1) = \mathbf{A_2} x_2(k) + \mathbf{B_2} u_2(k) + d_2(k),
\end{equation}
where $d_2(k) \in \mathbb{R}^{n_2}$ denotes additive system noise. In addition, suppose also that only the corrupted measurement/information is available:
    $\bar{x}_2(k) = x_2(k) + \Delta x_2(k)$, $\bar{u}_2(k) = u_2(k) + \Delta u_2(k)$, and $\mathbf{\bar{A}_1} = \mathbf{A_1} + \Delta \mathbf{A_1}$,
with $\Delta x_2(k)$ and $\Delta u_2(k)$ the measurement noise on the input and state, and $\Delta \mathbf{A_1}$ the knowledge mismatch on $\mathbf{A_1}$. Under this scenario, an \textit{empirical} estimation of the solution is given by
\begin{equation} \label{eq:PiPerturbed}
    \mathbf{\hat{\Theta}} := \bar{X}_{2, -} \hat{G}_{\Theta},
\end{equation}
with $\hat{G}_{\Theta} \in \RR^{T \times n_1}$ any matrix that satisfies
\begin{subequations} \label{eq:dataRepPiPerturbed} 
  \begin{empheq}[left=\empheqlbrace]{align}
    \bar{X}_{2, +} \hat{G}_{\Theta} &= \bar{X}_{2, -} \hat{G}_{\Theta} \mathbf{\bar{A}_1} \label{eq:dataRepPiPerturbed1} \\
    \bar{U}_{2, -}  \hat{G}_{\Theta} &=\mathbf{C_1}. \label{eq:dataRepPiPerturbed2}
  \end{empheq}
\end{subequations}
By defining an ``encapsulated'' noise matrix
\begin{equation} \label{eq:Rdef}
    R_{2,-} := \mathbf{A_2} \Delta X_{2,-} - \Delta X_{2,+} +\mathbf{B_2} \Delta U_{2,-} - D_{2,-}\,,
\end{equation}
a characterisation of the empirical error $\Delta \mathbf{\Theta} := \mathbf{\Theta} - \mathbf{\hat{\Theta}}$ and an error upper bound in terms of $R_{2,-}$ and $\Delta \mathbf{A_1}$ have been established in \cite{mao2025partOne}, which we recall as follows:

\begin{theorem} (see \cite[Lemma 3 and Theorem 4]{mao2025partOne}) \label{lemma:errorBound}
Suppose that the available (noisy) data matrices are such that
\begin{equation} \label{eq:rankX-U}
    \rank \left(\begin{bmatrix}
        \bar{X}_{2, -} \\ \bar{U}_{2, -} 
    \end{bmatrix}\right)
     = n_2 + m_2,
\end{equation}
and that $\lambda(\mathbf{A_1}) \cap \lambda(\mathbf{A_2}) = \emptyset$. Let $\hat{G}_{\Theta}$ be a solution of \eqref{eq:dataRepPiPerturbed}. Then, the following results hold.
\begin{enumerate}
    \item[(i)] $\Delta \mathbf{\Theta}$ is the unique solution of the Sylvester equation
\begin{equation} \label{eq:errorSyl}
    \mathbf{A_2} \Delta \mathbf{\Theta} -  \Delta \mathbf{\Theta} \mathbf{A_1} = - R_{2, -} \hat{G}_{\Theta} - \bar{X}_{2, -} \hat{G}_{\Theta} \Delta \mathbf{A_1}. 
\end{equation}
    \item[(ii)] The bound
        \begin{equation} \label{eq:boundPi}
            \|\Delta \mathbf{\Theta}\|_F
            \leq \frac{\|\hat{G}_{\Theta}\|_F \left( \| R_{2, -} \|_2 + \|\bar{X}_{2,-}\|_2 \|\Delta \mathbf{A_1}\|_2\right)}{\sigma_{min}(I \otimes \mathbf{A_2} - \mathbf{A_1}^\top \otimes I)} 
        \end{equation}
        holds.
\end{enumerate}
\end{theorem}

In Appendix~\ref{sec:extendedConfig}, for completeness, we investigate the extended configuration where $\mathbf{\Sigma_1}$ and $\mathbf{\Sigma_2}$ are both unknown. However, we note that this setting is rarely required in problems of interest. An exception is the static output regulation problem in Appendix~\ref{sec:regulationFullInfo}, where such an extension enables some additional results with respect to the literature. 

\subsection{Data-Driven Cascade Stabilisation} \label{sec:data-driven-cascade-stabilisation}
Consider the stabilisation problem of the following cascaded system
\begin{subequations} \label{eq:cascade}
    \begin{align}
    &x_1(k+1) = A_1x_1(k) + B_1u_1(k),\label{eq:cascade-x}\\ 
    &x_2(k+1) = A_2x_2(k) + B_2x_1(k) \label{eq:cascade-y}.
    \end{align}
\end{subequations} 
A trivial approach would be to design a stabilising state feedback by treating \eqref{eq:cascade} as a single monolithic system with the aggregated state $\mathcal{X}(k):= \col(x_1(k),x_2(k))$. While such a design philosophy is conceptually valid, it presents limitations in both theoretical and practical settings: (i) it disregards the known interconnection structure, which may have an impact on computational/design cost. For example, this is relevant in the dynamic output regulation problem where the cascade has a specific meaning: it is the interconnection of the system and an internal model. 
(ii) the approach is not scalable: when a new subsystem (\ie, an additional module) is appended to the cascade \eqref{eq:cascade}, the entire controller requires redesign. 

In contrast, the cascade stabilisation approach breaks the original (complex) stabilisation problem into lower-dimensional sub-problems while preserving the structure of the system, by stabilising the subsystems sequentially. This allows for reusable and sequential design for large problems. This approach consists of three meta steps, called \textit{forwarding method}, as recalled below.
\begin{enumerate}
    \item[(S1)] Pre-stabilise the first subsystem \eqref{eq:cascade-x} with a state feedback $u_1(k) = N_1 x(k) + v(k)$, where $N_1$ is such that $A_1 + B_1 N_1$ is Schur, and $v(k)$ is a new control input to be designed.
    \item[(S2)] Solve the Sylvester equation
    \begin{equation} 
    \label{eq:sylvester_cascade}
    A_2\Upsilon_c - \Upsilon_c(A_1 + B_1 N_1) = - B_2,
    \end{equation}
    which results from the observation that the pre-stabilised system \eqref{eq:cascade} mirrors the cascade interconnection $\mathbf{\Sigma_1} \rightarrow \mathbf{\Sigma_2}$, with $\mathbf{\Sigma_1} := (A_1 + B_1 N_1, B_1, I)$ and $\mathbf{\Sigma_2} := (A_2, B_2, I)$.
    \item[(S3)] Apply the coordinate transformation $\zeta := x_2 - \Upsilon_c x_1$ to the second subsystem \eqref{eq:cascade-y} and stabilising the transformed system 
    with $v(k) = N_2 \zeta(k)$, where $N_2$ is such that $A_2 - \Upsilon_c B_1 N_2$ is Schur.
\end{enumerate}
Once all these three steps have been executed, the cascade \eqref{eq:cascade}  in closed-loop with the overall control law $u_1(k) = N_1 x_1(k) + N_2 \zeta(k)$, yields that
\begin{equation} \label{eq:closed-loop-cascade-2-systems}
    \begin{bmatrix}
        x_1(k+1)\\\zeta(k+1)
    \end{bmatrix} = \begin{bmatrix}
        A_1 + B_1 N_1 & B_1 N_2\\ \mathbf{0} & A_2 - \Upsilon_c B_1 N_2
    \end{bmatrix}\begin{bmatrix}
        x_1(k)\\\zeta(k)
    \end{bmatrix},
\end{equation}
where the state matrix takes the upper-triangular form. As a consequence, the stability of each subsystem, characterised respectively by $A_1 + B_1 N_1$ and $A_2 - \Upsilon_c B_1 N_2$, implies the stability of the whole cascade system. 
Motivated by the advantages of the cascade stabilisation approach described above, a direct data-driven cascade stabilisation procedure has been proposed in Part I, bypassing the need for knowledge of the system matrices $A_1$, $B_1$, $A_2$ and $B_2$. For details of the data-driven implementations of each step of the forwarding method, we refer the reader to Section IV of Part I \cite{mao2025partOne}. We have recalled here the three steps of the forwarding method because in Section \ref{sec:outputRegulation} of this Part II, we will hinge on them steps conceptually to address the data-driven output regulation problem.

\color{black}

\section{Model Order Reduction} \label{sec:modelOrderReduction}
In this section, we solve the problem of data-driven model reduction by moment matching by casting it into the proposed interconnection framework. We begin by formulating the reduction procedure under the assumption that input-state measurements are available. Subsequently, we relax this assumption and develop an alternative approach that relies solely on input-output data.

\subsection{Problem Formulation}
Consider a single-input single-output linear time-invariant system of the form
\begin{equation}  \label{eq:sysReduction} 
    x(k+1) = Ax(k) + Bu(k), \quad y(k) = Cx(k),  
\end{equation}
with state $x(k) \in \mathbb{R}^{n}$, input $u(k) \in \mathbb{R}$, and output $y(k) \in \mathbb{R}$. Let 
$
W(z) = C(zI - A)^{-1}B, 
$
be the associated transfer function and assume that system~\eqref{eq:sysReduction} is minimal, \ie{}, it is both controllable and observable. The objective of model reduction by \textit{moment matching} is to determine a model or a family of models
\begin{equation} \label{eq:reducedOrderModel}
    \xi(k+1) = F\xi(k) + Gu(k), \quad \psi(k) = H\xi(k),  
\end{equation}
where $\xi(k) \in \RR^{\nu}$ and $\psi(k) \in \RR$, that interpolate the values of transfer function (known as \textit{moments}) of the \textit{full-order} system \eqref{eq:sysReduction} at prescribed complex values (known as \textit{interpolation points}). If $\nu < n$, then \eqref{eq:reducedOrderModel} is called a \textit{reduced-order} model. In this section, we use $\Sigma_{fom}$ to represent the full-order system \eqref{eq:sysReduction} and $\Sigma_{rom}$ for the reduced-order model, wherever they are convenient.

A formal definition of the moments of system $\Sigma_{fom}$ is given as follows. 

\begin{definition} \label{def:moment}
The $0$-moment of $\Sigma_{fom}$ at $z_i\in \mathbb{C} \setminus \lambda(A)$ is the complex number $\eta_0(z_i) := W(z_i)$. The $k$-moment of $\Sigma_{fom}$ at $z_i$ is the complex number $\eta_k(z_i):=\frac{(-1)^{k}}{k!} \left[ \frac{d^k}{dz^k} W(z)\right]_{z = z_i}$, where $k \geq 1$ is an integer. 
\end{definition}

The fundamental connection between moments and the solution of Sylvester equations was first identified in~\cite{gallivan2004sylvester, gallivan2006model}. This connection was further developed in~\cite{astolfi2010model}, as described below.

Let $\mathcal{I} = \{z_1,  \dots, z_{\overline{\nu}}\} \subset \mathbb{C} \setminus \lambda(A)$ be a set of interpolation points. Points in $\mathcal{I}$ may be repeated, thus we define $m_{z_i}$ as the multiplicity of the point $z_i$ in the set $\mathcal{I}$. Then, the moments $\eta_0(z_1)$, $\dots$, $\eta_{m_{z_1}-1}(z_1)$, $\dots$, $\eta_0(z_{\overline{\nu}})$, $\dots$, $\eta_{m_{z_{\overline{\nu}}}-1}(z_{\overline{\nu}})$ of $\Sigma_{fom}$ at $\mathcal{I}$ 
are in a \textit{one-to-one relation}\footnote{That is, there exists an invertible transformation between the two quantities.} with the entries of the matrix
    $C \Pi_m$, where $\Pi_m \in \mathbb{R}^{n \times \nu}$ is the unique solution to the Sylvester equation
    \begin{equation} \label{eq:SylEqPrimalReduction}
        A \Pi_m - \Pi_m S = -BL,
    \end{equation}
    provided that $S \in \mathbb{R}^{\nu \times \nu}$ is a non-derogatory\footnote{A matrix is non-derogatory if its characteristic and minimal polynomials coincide.} matrix with characteristic polynomial $\prod_{i=1}^{\overline{\nu}}(z-z_i)^{m_{z_i}}$, where $\sum_{i=1}^{\overline{\nu}} m_{z_i} = \nu $,
    and $L \in \mathbb{R}^{1 \times \nu}$ is a row vector such that the pair $(S, L)$ is observable\footnote{This result, originally developed for continuous-time systems, also holds in the discrete-time setting.}.



Due to the coordinate invariance of the moments (see~\cite{astolfi2010model}), the aforementioned one-to-one relation implies that the matrix $C \Pi_m$ can be viewed as an equivalent representation of the moments. For this reason,  as is common in the literature, we call the matrix $C\Pi_m$ ``moments (at $\mathcal{I}$)''. 

Given a set $\mathcal{I}$ with $\nu < n$, a family containing all reduced-order models of order $\nu$ matching the moments $C\Pi_m$ at $\mathcal{I}$ is given by
\begin{subequations} \label{eq:ROM_S}
\begin{align}
        \xi(k+1) &= (S - GL)\xi(k) + G u(k),  \\
        \quad \psi(k) &= C \Pi_m \xi(k),
\end{align}
\end{subequations}
for any matrix $G$ such that $\lambda(S) \cap \lambda(S - GL) = \emptyset$~\cite{astolfi2010model}. The matrix $G$ is a free parameter used to span the family to obtain models with additional desired properties. See \cite{astolfi2010model} for a detailed compilation of these additional properties.




\begin{remark} \label{remark:MOR-SG-interpreration}
The Sylvester equation $\eqref{eq:SylEqPrimalReduction}$ induces an interconnection interpretation of the moments $C\Pi_m$. Consider the ``signal generator''
\begin{equation*}
    \omega(k+1) = S\omega(k), \quad \theta(k) = L\omega(k),
\end{equation*}
which we denote as $\Sigma_{sg}$. 
The solution $\Pi_m$ of $\eqref{eq:SylEqPrimalReduction}$ characterises the invariant subspace
$
    \mathcal{M}_m = \{(x,\omega): x = \Pi_m \omega\}
$
of the cascade $\Sigma_{sg} \rightarrow \Sigma_{fom}$ via $u(k) = \theta(k)$. Then, it follows immediately that $C\Pi_m$ characterises the steady-state (provided it exists) output response of this cascade, namely $y_{ss}(k) = C\Pi_m \omega(k)$. This relation has been exploited by many related works, such as \cite{scarciotti2017data}.
\end{remark}

In the remainder of this section, we focus on determining the reduced-order model \eqref{eq:ROM_S} using directly either input-state or input-output data. 

\subsection{Reduction under Input-state Data} \label{sec:reductionISdata}
In this section we assume the availability of measurements of the state $x$ and of the input $u$. While the matrices $A$ and $B$ are assumed unknown, for the time being we consider the following assumption on $C$. 
\begin{assumption} \label{ass:knownBC}
    The output matrix $C$ is known a priori. 
\end{assumption}
\begin{remark}
  This assumption is valid in many \textit{grey-box} modelling scenarios, where each state variable has a well-defined physical meaning. For instance, in the dynamical model of a robotic manipulator, the state variable typically corresponds to either the position or velocity of joints. 
  In such cases, the output matrix $C$ is typically known in advance, provided that the sensor placements are specified during the design phase (\eg{}, like in the batch reactor model detailed in Section \ref{sec:exampleRegulation}). In Section~\ref{sec:reductionIOdata}, we remove this assumption and eliminate the reliance on state measurements, thereby extending the methodology to a more general data-driven setting.  
\end{remark} 

We note that the essential components to synthesise the reduced-order models in~\eqref{eq:ROM_S} are: (i) the pair $(S, L)$, which encodes the interpolation points, and (ii) the ``moments'' $C \Pi_m$. By Assumption~\ref{ass:knownBC}, the matrix $C$ is assumed to be known. Moreover, the matrix $S$ can be constructed directly from the prescribed set of interpolation points $\mathcal{I}$, for example using the real Jordan canonical form. Similarly, the vector $L$ can be trivially determined based on $S$ to satisfy the observability condition.

Given these, the central task for constructing the moment matching model \eqref{eq:ROM_S} consists in determining the solution $\Pi_m$ of the Sylvester equations~\eqref{eq:SylEqPrimalReduction}, based on available input-state data from the full-order system $\Sigma_{fom}$. This task can be 
readily achieved using Lemma~\ref{lemma:computePI} of the framework, as detailed next.

\begin{theorem} \label{theorem:ROM_SData}
Consider system $\Sigma_{fom}$ and a set $\mathcal{I} \subset \CC \setminus \lambda(A)$. Suppose that Assumption \ref{ass:knownBC} holds. Let $S \in \RR^{\nu \times \nu}$ be a non-derogatory matrix with characteristic polynomial $\prod_{i=1}^{\overline{\nu}}(z-z_i)^{m_{z_i}}$, where $\sum_{i=1}^{\overline{\nu}} m_{z_i} = \nu $,
and let $L \in \RR^{1 \times \nu}$ be such that $(S, L)$ is observable. Let $G \in \RR^{\nu}$ be such that $\lambda(S) \cap \lambda(S- GL) = \emptyset$. Suppose that the data from $\Sigma_{fom}$ is such that
\begin{equation} \label{eq:rankXUReduction}
    \rank \left(\begin{bmatrix}
        X_{-} \\ U_{-}
    \end{bmatrix}\right)
     = n + 1,
\end{equation} 
holds. Then, the model 
\begin{subequations} \label{eq:ROM_SData}
    \begin{align} 
    \xi(k+1) &= (S - GL)\xi(k) + G u(k), \\
    \psi(k) &= C X_{-} G_{\Pi_m} \xi(k), 
\end{align}
\end{subequations}
matches the moments of $\Sigma_{fom}$ at $\mathcal{I}$, with $G_{\Pi_m} \in \RR^{T \times \nu}$ any matrix that satisfies
\begin{subequations} \label{eq:dataRepPiReduction} 
  \begin{empheq}[left=\empheqlbrace]{align}
    X_{+}  G_{\Pi_m} &= X_{-} G_{\Pi_m} S,  \\
    U_{-}  G_{\Pi_m} &= L.    
  \end{empheq}
\end{subequations}
\end{theorem}

\begin{proof}
Observe that $\Pi_m = X_{-} G_{\Pi_m}$, which follows directly by applying Lemma~\ref{lemma:computePI} to the Sylvester equation \eqref{eq:SylEqPrimalReduction} and taking $\mathcal{A}_1 = S$, $\mathcal{C}_1 = L$, $\mathcal{A}_2 = A$, $\mathcal{B}_2 = B$ and $\mathbf{\Theta} = \Pi_m$. Substituting this expression into the model~\eqref{eq:ROM_SData} yields the model-based formulation in~\eqref{eq:ROM_S} and completes the proof.
\end{proof}

\subsection{Reduction under Input-output Measurements} \label{sec:reductionIOdata}
In Section~\ref{sec:reductionISdata}, the results were developed under Assumption~\ref{ass:knownBC} and with access to state measurements. In this section, we demonstrate that these requirements can be relaxed requiring only measurements of the input $u$ and output $y$, by leveraging the \textit{left difference operator representation}, see~\cite[Section 2.3.3]{goodwin2014adaptive} and~\cite[Section 6]{de2019formulas}. To this end, we represent system~\eqref{eq:sysReduction} as 
\begin{equation} \label{eq:differenceEq}
\begin{aligned}
&y(k) + a_{n} y(k-1) + \cdots + a_{2} y(k-n+1) + a_{1} y(k-n) \\
&\quad= b_{n} u(k-1) + \cdots + b_{2} u(k-n+1) + b_{1} u(k-n),
\end{aligned}
\end{equation}
where the coefficients $a_1, a_2, \dots, a_n$ and $b_1, b_2, \dots, b_n$ correspond to those in the transfer function $W(z)=\frac{Y(z)}{U(z)}$ associated with system~\eqref{eq:sysReduction}, namely
\begin{equation} \label{eq:tfReduction}
    \frac{Y(z)}{U(z)} = 
    \frac{b_n z^{n-1} + b_{n-1} z^{n-2} + \cdots + b_2 z + b_1}{z^n + a_n z^{n-1} + a_{n-1} z^{n-2} + \cdots + a_2 z + a_1}.
\end{equation}
By considering the augmented state
\begin{equation} \label{eq:xTilde}
\begin{gathered}
       \widetilde{x}(k) :=\col \bigl( y(k-n), y(k-n+1), \ldots, y(k-1) \\
u(k-n), u(k-n+1), \ldots, u(k-1) \bigr) \in \RR^{2n},
\end{gathered}
\end{equation}
we obtain the non-minimal state-space realisation $\widetilde{\Sigma}_{fom} := (\widetilde{A}, \widetilde{B}, \widetilde{C})$ defined in~\eqref{eq:stateSpaceBig}\footnote{This state-space form is precisely the one presented in~\cite[(58)]{de2019formulas}; we include it here for notational completeness.}.

\begin{figure*}[!t]
\normalsize
\begin{subequations}  \label{eq:stateSpaceBig}
\begin{align}
\widetilde{x}(k+1) &= 
\underbrace{
\begin{bmatrix}
    0 & 1 & 0 & \cdots & 0  & 0 & 0 & 0 & \cdots & 0 \\ 
    0 & 0 & 1 & \cdots & 0  & 0 & 0 & 0 & \cdots & 0 \\
    \vdots & \vdots & \vdots & \ddots & \vdots  & \vdots & \vdots & \vdots & \ddots & \vdots \\
    0 & 0 & 0 & \cdots & 1  & 0 & 0 & 0 & \cdots & 0 \\
    -a_1  & -a_2 & -a_3 & \cdots & -a_n  & b_1 & b_2 & b_3 & \cdots & b_{n} \\
    \hdashline
    0 & 0 & 0 & \cdots & 0  & 0 & 1 & 0 & \cdots & 0 \\ 
    0 & 0 & 0 & \cdots & 0  & 0 & 0 & 1 & \cdots & 0 \\ 
    \vdots & \vdots & \vdots & \ddots & \vdots  & \vdots & \vdots & \vdots & \ddots & \vdots \\
    0 & 0 & 0 & \cdots & 0  & 0 & 0 & 0 & \cdots & 1 \\ 
    0 & 0 & 0 & \cdots & 0  & 0 & 0 & 0 & \cdots & 0 \\ 
\end{bmatrix}
}_{\displaystyle \widetilde{A}}
\widetilde{x}(k)
+
\underbrace{
\begin{bmatrix}
    0 \\
    0 \\ 
    \vdots \\
    0 \\
    0 \\
    \hdashline
    0 \\
    0 \\
    \vdots \\
    0 \\
    1
\end{bmatrix}
}_{\displaystyle \widetilde{B}}
u(k), \\
y(k) &= 
\underbrace{
\begin{bmatrix}
    -a_1  & -a_2 & -a_3 & \cdots & -a_n  & b_1 & b_2 & b_3 & \cdots & b_{n}
\end{bmatrix}
}_{\displaystyle \widetilde{C}}
\widetilde{x}(k).
\end{align}
    \end{subequations}
\hrulefill
\vspace*{4pt}
\end{figure*}

The core objective is to recover the matrix $C \Pi_m$ from input-output data. To this end, we first establish that the ``moments'' $C \Pi_m$ of the original system~\eqref{eq:sysReduction} are equal to those of the non-minimal realization~\eqref{eq:stateSpaceBig}, provided that $0$ is not an interpolation point (\ie, $0 \notin \lambda(S)$).

\begin{lemma} \label{lemma:invarianceofMoments}
Consider system $\Sigma_{fom}$ and the representation $\widetilde{\Sigma}_{fom}$. Suppose that $\lambda(A) \cap \lambda(S) = \emptyset$, with $0 \notin \lambda(S)$, and that the pair $(S, L)$ is observable. Then, it holds that
\begin{equation} \label{eq:momentInvariant}
    C\Pi_m = \widetilde{C} \widetilde{\Pi}_m,
\end{equation}
with $\widetilde{\Pi}_m \in \RR^{2n \times \nu}$ the unique solution of
\begin{equation} \label{eq:SylEqTilde}
    \widetilde{A}\widetilde{\Pi}_m - \widetilde{\Pi}_m S= -\widetilde{B} L, 
\end{equation}
which is the Sylvester equation associated with $\widetilde{\Sigma}_{fom}$.
\end{lemma}
\begin{proof}
The relation \eqref{eq:momentInvariant} follows trivially by recalling that $\Sigma_{fom}$ and $\widetilde{\Sigma}_{fom}$ have the same transfer function \eqref{eq:tfReduction}. Consequently, they have the same moments at all well-defined interpolation points.\\ 
To establish the uniqueness of $\widetilde{\Pi}$, observe that $\Pi$ is unique and that $0 \notin \lambda(S)$. 
By \eqref{eq:stateSpaceBig}, $\widetilde{A}$ is a (upper) block-triangular matrix, whose diagonal consists of two blocks.
\begin{itemize}
  \item The bottom-right block is an upper triangular matrix with all diagonal entries equal to $0$, hence it has a single eigenvalue at $0$ with algebraic multiplicity $n$.
  \item The upper-left block is the companion matrix
  $$
  \begin{bmatrix}  
  0 & 1 & 0 & \cdots & 0 \\
  0 & 0 & 1 & \cdots & 0 \\
  \vdots & \vdots & \vdots & \ddots & \vdots \\
  0 & 0 & 0 & \cdots & 1 \\
  -a_1 & -a_2 & -a_3 & \cdots & -a_n  
  \end{bmatrix},
  $$
  whose characteristic polynomial is $p(\lambda) = \lambda^n + a_n \lambda^{n-1} + \cdots + a_2 \lambda + a_1$, thereby its eigenvalues coincide with those of $A$ (see~\cite[Section 3.3]{horn2012matrix}).
\end{itemize}
Thus, $\lambda(\widetilde{A})$ is the union of $\lambda(A)$ and $\{0\}$ and, by assumption, $\lambda(S)$ does not contain $0$ and is disjoint from $\lambda(A)$. Therefore, $\lambda(\widetilde{A}) \cap \lambda(S) = \emptyset$, which implies that \eqref{eq:SylEqTilde} has a unique solution. This completes the proof.
\end{proof}

\begin{remark}
The requirement that $0$ is not an interpolation point is not restrictive for \emph{discrete-time} moment matching reduction methods. In fact, for discrete-time systems, the interpolation points are typically selected to be on the unit circle, \ie{}, $z_i = e^{\iota \omega_i t_s}$ with $\iota$ the imaginary unit, $t_s$ the sampling time and $\omega_i \in \RR$ the frequency of interest.
\end{remark}

Lemma~\ref{lemma:invarianceofMoments} shows that, instead of computing the solution $\Pi_m$ to the Sylvester equation~\eqref{eq:SylEqPrimalReduction}, which depends on state measurements $x$, one can alternatively compute the solution $\widetilde{\Pi}_m$ to~\eqref{eq:SylEqTilde}, which is based solely on the augmented state $\widetilde{x}$ and thus requires only input-output measurements.

However, a challenge, namely that the matrix $\widetilde{C}$ is unknown, remains. As a matter of fact, knowing $\widetilde{C}$ would imply knowledge of all system parameters $a_1, a_2, \dots, a_n$ and $b_1, b_2, \dots, b_n$, which is equivalent to system identification and so contrary to the motivation of our direct data-driven framework.

To resolve this challenge, we observe that the representation $\widetilde{\Sigma}_{fom}$ in \eqref{eq:stateSpaceBig} is in a highly structured canonical form. Thus, one wonders whether the associated $\widetilde{\Pi}_m$ may also have a particular structure. This is indeed the case, as shown in the following result.

\begin{lemma} \label{lemma:PitildeStructure}
Consider the representation $\widetilde{\Sigma}_{fom}$. Suppose that $\lambda(A) \cap \lambda(S) = \emptyset$, with $0 \notin \lambda(S)$. Let $\widetilde{\Pi}_m$ be the unique solution of \eqref{eq:SylEqTilde} and let $\widetilde{\Pi}_{m, j}$ denote its $j$-th row. 
Then, 
\begin{equation} \label{eq:PiTildeStructure}
    \widetilde{C}\widetilde{\Pi}_m = \widetilde{\Pi}_{m, j} S^{n-j+1},
\end{equation}
for any $j = 1, 2, \dots, n$.  
\end{lemma}
\begin{proof}
In this proof, we reveal the special structure of $\widetilde{\Pi}_m$ by resorting to the interconnection-based interpretation of the Sylvester solution, as noted in Remark \ref{remark:MOR-SG-interpreration}. Consider the cascade $\Sigma_{sg} \rightarrow \widetilde{\Sigma}_{fom}$. We know that for this interconnection there exists an invariant subspace $\widetilde{\mathcal{M}}_m = \{(\widetilde{x}, \omega): \widetilde{x} = \widetilde{\Pi}_m \omega\}$.
Then, on this subspace $\widetilde{\mathcal{M}}_m$, the augmented state $\widetilde{x}(k)$ must satisfy
\begin{align*}
\widetilde{x}&(k) := \\
& 
\begin{bmatrix}
    \widetilde{C}\widetilde{x}(k-n) \\
    \vdots \\
    \widetilde{C}\widetilde{x}(k-1) \\
    u(k-n) \\
    \vdots \\
    u(k-1)
\end{bmatrix} 
= 
\begin{bmatrix}
    \widetilde{C}\widetilde{\Pi}_m \omega(k-n) \\
    \vdots \\
    \widetilde{C}\widetilde{\Pi}_m \omega(k-1) \\
    L\omega(k-n) \\
    \vdots \\
    L\omega(k-1)
\end{bmatrix} 
= 
\begin{bmatrix}
    \widetilde{C}\widetilde{\Pi}_m S^{-n} \\
    \vdots \\
    \widetilde{C}\widetilde{\Pi}_m S^{-1} \\
    L S^{-n} \\
    \vdots \\
    L S^{-1}
\end{bmatrix}
\omega(k),
\end{align*}
where in the second equality we used the identity $\omega(k-j) = S^{-j} \omega(k)$ for $j = 1, \dots, n$ (which holds by the invertibility of $S$).
On the other hand, from the definition of the invariant subspace, we also have 
$\widetilde{x}(k) = \widetilde{\Pi}_m \omega(k)$. 
Comparing the two expressions for $\widetilde{x}(k)$, we obtain that the rows of $\widetilde{\Pi}_m$ satisfy
\begin{equation}
\widetilde{\Pi}_{m,1} = \widetilde{C} \widetilde{\Pi}_m S^{-n}, \quad \dots, \quad \widetilde{\Pi}_{m,n} = \widetilde{C} \widetilde{\Pi}_m S^{-1}.
\end{equation}
Finally, post-multiplying both sides of these equations by the corresponding powers of $S$ yields the structure stated in equation~\eqref{eq:PiTildeStructure}.
\end{proof}

Lemma~\ref{lemma:PitildeStructure} allows computing $\widetilde{C} \widetilde{\Pi}_m$ directly from $\widetilde{\Pi}_m$, avoiding the requirement on any knowledge of $\widetilde{C}$. At this point we can determine a reduced-order model using only input-output data. 

\begin{theorem} \label{theorem:ROMSDataIO}
Consider system $\Sigma_{fom}$ and a set $\mathcal{I} \subset \CC \setminus (\lambda(A) \cup \{0\})$. Let $S \in \RR^{\nu \times \nu}$ be a non-derogatory matrix with characteristic polynomial $\prod_{i=1}^{\overline{\nu}}(z-z_i)^{m_{z_i}}$, where $\sum_{i=1}^{\overline{\nu}} m_{z_i} = \nu $,
and let $L \in \RR^{1 \times \nu}$ be such that the pair $(S, L)$ is observable. Let $G \in \RR^{\nu}$ be such that $\lambda(S) \cap \lambda(S- GL) = \emptyset$, and 
let $e_n \in \RR^{1 \times 2n}$ be a vector with all zero elements apart for the element in position $n$ which is equal to $1$. Suppose that the data from $\Sigma_{fom}$ is such that
\begin{equation} \label{eq:rankXUReduction-tilde}
    \rank \left(\begin{bmatrix}
        U_{-} \\ \widetilde{X}_{-}
    \end{bmatrix}\right)
     = 2n + 1,
\end{equation} 
holds. Then, the model 
\begin{subequations} \label{eq:ROM_SDataIO}
    \begin{align} 
    \xi(k+1) &= (S - GL)\xi(k) + G u(k), \\
    \psi(k) &= e_n \widetilde{X}_{-} G_{\widetilde{\Pi}_m} S \xi(k),  \label{eq:ROM_SDataIO2}
\end{align}
\end{subequations}
matches the moments of $\Sigma_{fom}$ at $\mathcal{I}$, with $G_{\widetilde{\Pi}_m} \in \RR^{T \times \nu}$ any matrix that satisfies
\begin{subequations} \label{eq:dataRepPiReductionIO} 
  \begin{empheq}[left=\empheqlbrace]{align}
    \widetilde{X}_{+}  G_{\widetilde{\Pi}_m} &= \widetilde{X}_{-} G_{\widetilde{\Pi}_m} S,  \\
    U_{-}  G_{\widetilde{\Pi}_m} &= L. 
  \end{empheq}
\end{subequations}    
\end{theorem}

\begin{proof}
Observe first that the term $\widetilde{X}_{-} G_{\widetilde{\Pi}_m}$ in equation~\eqref{eq:ROM_SDataIO2} evaluates to $\widetilde{\Pi}_m$. This follows by applying Lemma~\ref{lemma:computePI} to \eqref{eq:SylEqTilde}, taking $\mathcal{A}_1 = S$, $\mathcal{C}_1 = L$, $\mathcal{A}_2 = \widetilde{A}$, $\mathcal{B}_2 = \widetilde{B}$ and $\mathbf{\Theta} = \widetilde{\Pi}_m$. 
Next, by pre-multiplying by the vector $e_n$, we extract the $n$-th row of $\widetilde{\Pi}_m$. Post-multiplying this by $S$ yields $\widetilde{C}\widetilde{\Pi}_m$, as established in Lemma~\ref{lemma:PitildeStructure}, and in turn $C \Pi_m$ by Lemma \ref{lemma:invarianceofMoments}. Therefore, model~\eqref{eq:ROM_SDataIO} is the reduced-order model by moment matching given in \eqref{eq:ROM_S}. This completes the proof.
\end{proof}

\begin{remark}
An alternative line of research in data-driven moment matching is based on the physical embodiment of the signal generator $\Sigma_{sg}$, see, \eg, \cite{scarciotti2017data, mao2024data, mao2024data-Simu, zhang2025data}. These methods approximate the ``moments'' $C\Pi_m$ from the near-steady-state data from the cascade $\Sigma_{sg} \rightarrow \Sigma_{fom}$, as noted in Remark~\ref{remark:MOR-SG-interpreration}. We highlight that the approach proposed in this section is exact rather than approximate. Moreover, our approach presents several additional advantages: (i) it does not require to run the experiment $\Sigma_{sg} \rightarrow \Sigma_{fom}$, thereby allowing data to be collected during normal operation rather than the extraordinary operation of actually driving the system with a signal generator; (ii) no stability assumptions are required for the system, since our method does not rely on steady-state responses; (iii) interpolation points are not restricted to being on the unit circle; (iv) high-order moments can be treated. Note that item (iv) (and (iii) when the interpolation points are outside the unit circle) remains an unsolved challenge for steady-state-based approaches because the corresponding signal generators would produce diverging inputs (and, thus, the steady state upon which the methods are based might not exist).
\end{remark}

The results of this section are summarised in a computational procedure outlined in Algorithm~\ref{alg:dataDrivenModelReduction}.

\begin{algorithm}
  \caption{Data-Driven Reduction by Moment Matching}
  \label{alg:dataDrivenModelReduction}
  \begin{algorithmic}[1]
    \STATE \textbf{Input:} Collected samples of $u$ and $y$ from $\Sigma_{fom}$, $\mathcal{I}$ of cardinality less than $n$, prescribed $G$ to enforce additional properties (\eg, stability) 
    \STATE \textbf{\# \, Assemble data matrices}
    \STATE \hspace{1em} Construct matrices $S$ and $L$
    \STATE \hspace{1em} Construct data matrices $(U_-, \widetilde{X}_-, \widetilde{X}_+)$
    \STATE \textbf{\# \, Solve Sylvester equation from data}
    \STATE \hspace{1em} Solve the LP \eqref{eq:dataRepPiReductionIO} to obtain $G_{\widetilde{\Pi}_m}$
    \STATE \hspace{1em} Obtain $\widetilde{\Pi}_m = \widetilde{X}_{-} G_{\widetilde{\Pi}_m}$ 
    \STATE  \textbf{\# \, Recover the ``moments'' of $\Sigma_{fom}$}
    \STATE \hspace{1em} Obtain $C\Pi_m = e_n \widetilde{\Pi}_m S$
    \STATE \textbf{Output:} Reduced-order model \eqref{eq:ROM_SDataIO} that interpolates the moments of $\Sigma_{fom}$ 
    at $\mathcal{I}$
  \end{algorithmic}
\end{algorithm}

\subsection{Moment Mismatch under Noise}
Suppose now that the available input and output measurements are subject to (additive) noise $\Delta u$ and $\Delta y$, respectively. Moreover, we allow for the dynamics of the linear system to be affected by a process noise $d$. Such process noise could, for instance, arise when the higher order, unknown model is actually nonlinear. In such a setting,
$d$ may be intended as a container for the nonlinearities. This leads to the reduced-order model
\begin{subequations} \label{eq:ROM_SDataIONoise}
    \begin{align} 
    \xi(k+1) &= (S - GL)\xi(k) + G u(k), \\
    \psi(k) &= e_n \bar{X}_{-} \hat{G}_{{\Pi}_m} S \xi(k),  \label{eq:ROM_SDataIO2Noise}
\end{align}
\end{subequations}
with $\hat{G}_{{\Pi}_m} \in \RR^{T \times \nu}$ any matrix that satisfies
\begin{subequations} \label{eq:dataRepPiReductionIONoise} 
  \begin{empheq}[left=\empheqlbrace]{align}
    \bar{X}_{+}  \hat{G}_{{\Pi}_m} &= \bar{X}_{-} \hat{G}_{{\Pi}_m} S,  \\
    \bar{U}_{-}  \hat{G}_{{\Pi}_m} &= L, 
  \end{empheq}
\end{subequations}    
where $\bar{X}_{-} = \widetilde{X}_{-} + \Delta \widetilde{X}_{-}$ and $\bar{X}_{+} = \widetilde{X}_{+} + \Delta \widetilde{X}_{+}$ are the noisy measurement matrices. It is expected that \eqref{eq:ROM_SDataIONoise} cannot match precisely the moments of \eqref{eq:sysReduction} at $\mathcal{I}$ due to the influence of noise. Therefore, in the next result we obtain a bound on such moment mismatch by applying Theorem~\ref{lemma:errorBound}. Before doing so, we note that the corresponding measurement noise for the aggregated state is expressed in terms of $\Delta y$ and $\Delta u$ as
\begin{equation*}
\begin{gathered}
       \Delta \widetilde{x}(k) =\col \bigl(\Delta y(k-n), \Delta y(k-n+1), \ldots, \Delta y(k-1) \\
\Delta u(k-n), \Delta u(k-n+1), \ldots, \Delta u(k-1) \bigr),
\end{gathered}
\end{equation*}
based on which we define the ``encapsulated'' noise signal with respect to $\widetilde{\Sigma}_{fom}$ as
\begin{equation*}
    \widetilde{r}(k) := \widetilde{A} \Delta \widetilde{x}(k) - \Delta \widetilde{x}(k+1) + \widetilde{B} \Delta u(k) - \widetilde{d}(k),
\end{equation*}
and its corresponding data matrix version
$$
\widetilde{R}_{-} := \widetilde{A} \Delta \widetilde{X}_{-} - \Delta \widetilde{X}_{+} + \widetilde{B} \Delta U_- - \widetilde{D}_-.
$$

\begin{theorem}
Consider the reduced-order model \eqref{eq:ROM_SDataIONoise} and the associated transfer function $\widehat{W}(z)$. Let $\mathcal{T}$ be the unique matrix such that $S = \mathcal{T}^{-1} S_{J} \mathcal{T}$ and $L = L_{J} \mathcal{T}$, with $S_{J}$ the Jordan canonical form of $S$ and $L_{J}$ the corresponding Jordan block selector row. 
Then
\begin{equation} \label{eq:momentMismatchBound}
\begin{split}
    |\eta_j(z_i) - \widehat{\eta}_j(z_i) |   
    \leq   \frac{\|S\|_2 \|\mathcal{T}^{-1}\|_2 \|\hat{G}_{{\Pi}_m}\|_F}{\sigma_{min}(I \otimes \widetilde{A} - S^\top \otimes I)} \| R_{-} \|_2 
\end{split}
\end{equation}
holds for all $z_i \in \mathcal{I}$ and $j< m_{z_i}$, where $\widehat{\eta}_j(z_i)$ is the $j$-moment at $z_i$ corresponding to $\widehat{W}(z)$.

\end{theorem}

\begin{proof}
In what follows we prove that $\max_{z_i \in \mathcal{I}}  |\eta_j(z_i) - \widehat{\eta}_j(z_i) |$ is subject to the bound in \eqref{eq:momentMismatchBound}. To this end, define $\Delta \eta_j(z_i) := \eta_j(z_i) - \widehat{\eta}_j(z_i)$, and note that the expression $\max_{z_i \in \mathcal{I}}  | \eta_j(z_i) - \widehat{\eta}_j(z_i) |$ is equivalent to the norm
\begin{equation*}
         \| \vecop (\begin{bmatrix}
             \Delta \eta_0(z_1) \!\!&\!\! \dots \!\!&\!\! \Delta \eta_{m_{z_0}-1}(z_1) \!\!&\!\! \dots \!\!&\!\! \Delta \eta_{m_{z_{\overline{\nu}}}-1}(z_{\overline{\nu}})
         \end{bmatrix}) \|_{\infty},
\end{equation*}
which is upper bounded by 
\begin{equation} \label{eq:momentMismatchF}
         \| \begin{bmatrix}
             \Delta \eta_0(z_1) \!\!&\!\! \dots \!\!&\!\! \Delta \eta_{m_{z_0}-1}(z_1) \!\!&\!\! \dots \!\!&\!\! \Delta \eta_{m_{z_{\overline{\nu}}}-1}(z_{\overline{\nu}})
         \end{bmatrix} \|_{F}.
\end{equation}
We recall that \cite{astolfi2010model} established that the ordered (with respect to $S_j$) collection of moments listed in a row vector multiplied to the right by $\mathcal{T}$ is equal to
the matrix $C\Pi_m$. Then exploiting the identity between $C\Pi_m$ and $\widetilde{C} \widetilde{\Pi}_m$ established in Lemma~\ref{lemma:invarianceofMoments}, yields that \eqref{eq:momentMismatchF} is equal to
\begin{equation*} 
         \| (\widetilde{C} \widetilde{\Pi}_m - e_n \bar{X}_{-} \hat{G}_{{\Pi}_m} S) \mathcal{T}^{-1} \|_{F} = \| e_n(\widetilde{\Pi}_{m} - \widehat{\Pi}_{m}) S \mathcal{T}^{-1} \|_{F},
\end{equation*}
where, by Lemma~\ref{lemma:PitildeStructure}, $e_n \widehat{\Pi}_m S$, with $\widehat{\Pi}_m := \bar{X}_{-} \hat{G}_{{\Pi}_m}$, are the ``noisy'' moments defined via \eqref{eq:dataRepPiReductionIONoise}, while the real moments satisfy $\widetilde{C} \widetilde{\Pi}_m = e_n \widetilde{\Pi}_{m} S$ by the same lemma. 
Applying the mixed-norm inequality\footnote{Given matrices \(\mathcal{A}\) and \(\mathcal{B}\) of compatible dimensions, the inequalities \(\|\mathcal{A} \mathcal{B}\|_F \leq \|\mathcal{A}\|_2 \|\mathcal{B}\|_F\) and \(\|\mathcal{A} \mathcal{B}\|_F \leq \|\mathcal{A}\|_F \|\mathcal{B}\|_2\) hold.}, the above expression is upper-bounded by
\begin{equation*} 
         \| e_n(\widetilde{\Pi}_{m} - \widehat{\Pi}_{m}) \|_{F} \|S\|_2 \|\mathcal{T}^{-1}\|_2,
\end{equation*}
and, by the row-selection role of $e_n$, further upper bounded by
\begin{equation} \label{eq:momentMismatchBound1}
         \| \widetilde{\Pi}_{m} - \widehat{\Pi}_{m} \|_{F} \|S\|_2 \|\mathcal{T}^{-1}\|_2.
\end{equation}
Note now that the bound
\begin{equation} \label{eq:momentMismatchBound2}
    \begin{split}
    \|\widetilde{\Pi}_{m} - \widehat{\Pi}_{m}\|_{F}  
    \leq   \frac{\|\hat{G}_{\Pi_m}\|_F}{\sigma_{min}(I \otimes \widetilde{A} - S^\top \otimes I)} \| R_{-} \|_2, 
\end{split}
\end{equation}
follows directly from Theorem~\ref{lemma:errorBound} for $\Delta \mathbf{A_1} = \mathbf{0}$ (as $S$ is constructed from the user's selected interpolation points, hence it is known exactly). Last, combining \eqref{eq:momentMismatchBound1} and \eqref{eq:momentMismatchBound2} yields the claimed bound in \eqref{eq:momentMismatchBound}, which completes the proof.
\end{proof}


\subsection{Numerical Example}
Consider the one-dimensional heat equation model \cite{morChaV92}, a widely used benchmark for model order reduction. This heat equation describes the temperature field on a thin rod, which is spatially discretised into $201$ segments (leading to $n=200$), under an external heating (\ie, the input) at $1/3$ of the length. The output is the temperature recorded at $2/3$ of the length. To illustrate the use of the proposed method outlined in Theorem~\ref{theorem:ROMSDataIO} (and the corresponding Algorithm~\ref{alg:dataDrivenModelReduction}), we (temporally) discretise the given continuous-time model using a sampling time $t_s = 0.1$~s. 

Consider a set of frequencies of interests, namely $\omega_1 = 0.15$~rad/s, $\omega_2= 4.83$~rad/s and $\omega_3 = 31.62$~rad/s. Suppose that one wants to interpolate the $0$-moments at all these frequencies, and interpolate the $1$-moment at frequency $\omega_2$. This selection leads to the set of complex points
$\mathcal{I}_z := \{e^{\iota 0.15 t_s}, \; e^{\iota 4.83 t_s}, \; e^{\iota 31.62 t_s} \}$ associated with multiplicities $[1, 2, 1]$,
and their complex conjugates $\mathcal{I}_z^*$. Hence, the interpolation set $\mathcal{I} := \mathcal{I}_{z} \cup \mathcal{I}_z^*$ consists of $8$ interpolation points of modulus one, constituting $4$ complex conjugate pairs. This allows the matrix $S$ to be constructed in the real Jordan form. The vector $L$ is constructed such that $(S, L)$ is observable. The matrix $G$ is selected such that the eigenvalues of the reduced-order model \eqref{eq:ROM_SDataIO} are placed inside the open unit circle (which is always possible by the observability of the pair $(S, L)$). We collect an input-output data sequence $(u, y)$ of length $T = 803$. We then solve the LP in \eqref{eq:dataRepPiReductionIO} using MATLAB CVX \cite{cvx} with MOSEK \cite{aps2019mosek}, ultimately leading to a reduced-order model of order $8$. 

Table~\ref{tab:momentMismatch} reports the moments of both the full-order and reduced-order models associated with the set $\mathcal{I}_z$. The moments corresponding to the conjugate set $\mathcal{I}_z^*$ are omitted, as they yield identical results due to complex conjugate symmetry. As shown in the table, the reduced-order model matches the moments of the original system with high numerical accuracy, thus confirming the established interpolation property.

Fig.~\ref{fig:MOR_bodes} displays the Bode plots of the original full-order system (solid/blue line) and the reduced-order model (dashed/green line). The interpolation points associated with $\mathcal{I}$ are represented by red triangles ($0$-moments, \ie{}, with multiplicity of $1$) and magenta squares (both $0$- and $1$-moments, \ie{}, with multiplicity of $2$), respectively, showing consistency with Table~\ref{tab:momentMismatch}. 
Despite having a dimension equal to only 4\% of the original system, the reduced model captures the frequency-domain response of the original system with excellent fidelity across a broad frequency range, rather than just in a neighbourhood of the interpolation points. This demonstrates the effectiveness of the proposed model reduction approach.


\begin{table}[!t]
  \caption{Moments at the Selected Interpolation Points}          
  \centering
  \begin{tabular}{c r r}
    \toprule
      \textbf{Moment} & \textbf{Full-order model} & \textbf{Reduced-order model} \\
    \midrule
      $\eta_0(e^{\iota 0.15 t_s})$   & $(5.02 - 28.64\iota)10^{-3}$ & $(5.02 - 28.64\iota)10^{-3}$ \\
      $\eta_0(e^{\iota 4.83 t_s})$   & $(6.72 + 1.05\iota) 10^{-5}$ & $(6.72 + 1.04\iota) 10^{-5}$ \\
      $\eta_1(e^{\iota 4.83 t_s})$   & $(2.32 + 5.36\iota) 10^{-4}$ & $(2.32 + 5.37\iota) 10^{-4}$ \\
      $\eta_0(e^{\iota 31.62 t_s})$  & $(-1.13 + 0.08\iota) 10^{-8}$ & $(-1.13 + 0.09\iota) 10^{-8}$ \\
    \bottomrule
  \end{tabular}
  \label{tab:momentMismatch}
\end{table}

\begin{figure}[ht]
    \centering
    \includegraphics[width=0.95\columnwidth]{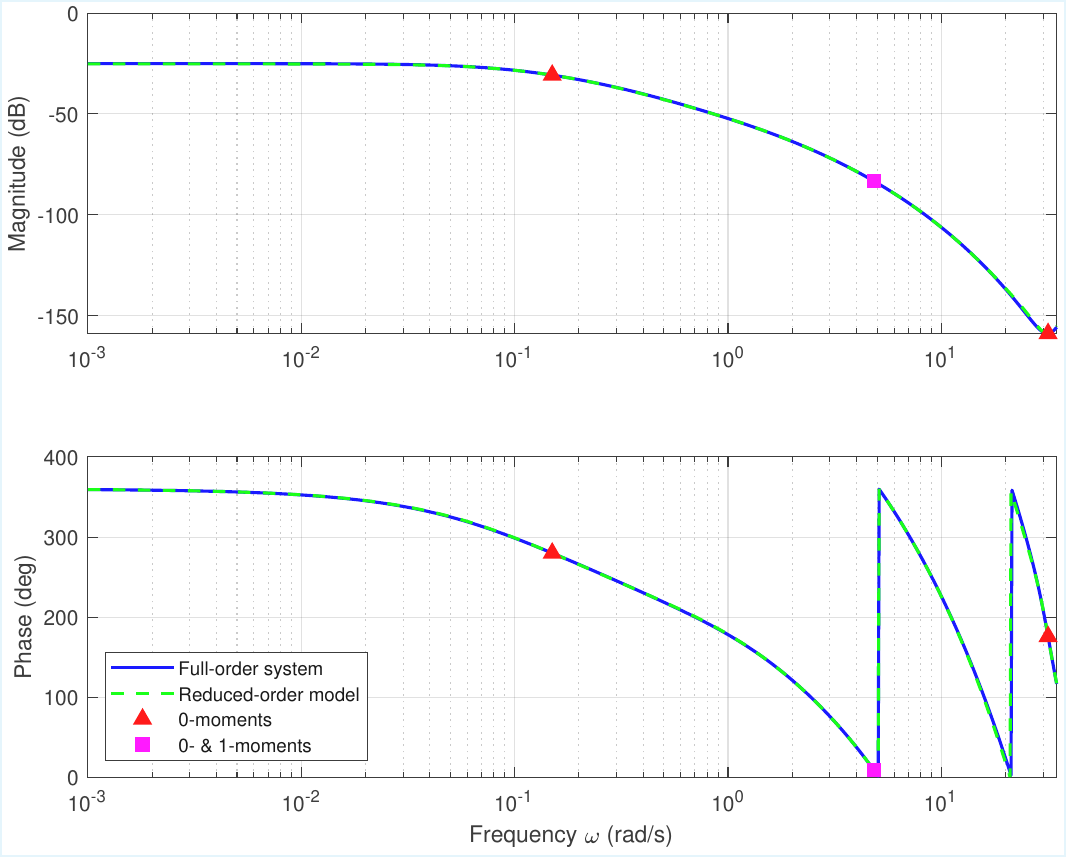}
    \caption{Bode plot of the full-order system 
    (solid/blue line) and of the reduced-order model (dashed/green line). The red triangles denote the points where $0$-moments are interpolated, and the magenta square denotes the point where both $0$-moments and $1$-moments are interpolated.}
    \label{fig:MOR_bodes}
\end{figure}

\section{Output Regulation} \label{sec:outputRegulation}
In this section, we show how data-driven output regulation problems can be recast in the proposed interconnection framework.  
Two data-driven problems are formulated. The solution of the first (presented in Appendix~\ref{sec:regulationFullInfo}) hinges upon the so-called regulator equations, whereas that of the second is based on the internal model principle. 

\subsection{Problem Formulation} \label{sec:output-regulation-formulation}
Consider the interconnected system
\begin{subequations}  \label{eq:sysRegulation}
\begin{align} 
    \omega(k+1) &= S \omega(k), \label{eq:exosystem} \\
    x(k+1) &= Ax(k) + Bu(k) + E\omega(k), \label{eq:sysXRegulation} \\
    e(k) &= Cx(k) + Du(k) + F\omega(k), \label{eq:sysERegulation}
\end{align}  
\end{subequations} 
where \( x(k) \in \mathbb{R}^{n} \) denotes the state of the plant, \( u(k) \in \mathbb{R}^{m} \) is the control input, \( e(k) \in \mathbb{R}^p \) is the regulation error, and \( \omega(k) = \col(r(k), d(k)) \in \mathbb{R}^{\nu} \) is the \textit{exogenous signal}, which encapsulates both reference trajectories $r$ to be tracked and disturbances $d$ to be rejected. The signal \( \omega \) is generated by the so-called \emph{exosystem} \eqref{eq:exosystem}, where the matrix $S$ is assumed known and to have no eigenvalues strictly inside the unit circle of the complex plane\footnote{The latter assumption is standard in order to ensure that the components of $\omega$ do not trivially converge to 0, see also \cite[Assumption 1.6]{huang2004nonlinear}.}. For the sake of notational convenience, we use $\Sigma_{exo} := (S)$ to denote the exosystem \eqref{eq:exosystem}, $\Sigma_{plant} := (A,B,C,D,E,F)$ to denote the plant \eqref{eq:sysXRegulation}-\eqref{eq:sysERegulation}, and $\Sigma_{exo} \rightarrow \Sigma_{plant}$ to denote the cascade system \eqref{eq:sysRegulation}.

The output regulation problem (see \cite{huang2004nonlinear}) consists in designing a regulator $u$ such that the following two objectives are achieved.
\begin{itemize}
    \item[(O1)] The cascade system $\Sigma_{exo} \rightarrow \Sigma_{plant}$ in closed-loop with the regulator under $\omega \equiv 0$ is exponentially stable.
    \item[(O2)]{The cascade system $\Sigma_{exo} \rightarrow \Sigma_{plant}$ in closed-loop with the regulator satisfies
    $
        \lim_{k \to \infty} e(k) = 0, 
    $
    for any initial conditions $x(0)$ and $\omega(0)$. } 
\end{itemize}
Bearing these two objectives in mind, we can formulate the following two data-driven (\ie, no knowledge on the matrices associated with the tuple $\Sigma_{plant}$) output regulation problems, depending on what data is available. 

\textit{\textbf{Problem 1} (Data-driven Static State-feedback Problem}): 
Design a regulator of the form
\begin{equation*}
    u = u_{stat}(x, \omega),
\end{equation*}
using data from the cascade system \eqref{eq:sysRegulation}, such that objectives (O1) and (O2) are achieved. \smallskip

\textit{\textbf{Problem 2} (Data-driven Dynamic State-feedback Problem}):  
Design a regulator of the form
\begin{equation*}
    u = u_{dyna}(x, \xi),
\end{equation*}
where $\xi$ is an auxiliary signal produced by the so-called \textit{internal model} unit \cite{francis1976internal}
\begin{equation} \label{eq:internalModel}
    \Sigma_{imo}:\,\, \xi(k+1) = \Phi\xi(k) + \Psi e(k), 
\end{equation}
using data from the cascade system \eqref{eq:sysRegulation},
such that objectives (O1) and (O2) are achieved. \smallskip

From standard output regulation theory, it is known that 
the matrices $\Phi$ and $\Psi$ are given by
\begin{equation} \label{eq:p-internal-model}
    \Phi = \bdiag \underbrace{(\phi, \dots, \phi)}_{p \text{ times}}, \quad \Psi = \bdiag \underbrace{(\psi, \dots, \psi)}_{p \text{ times}},
\end{equation}
where $\phi \in \RR^{q \times q}$ is a matrix whose characteristic polynomial coincides with the minimal polynomial of $S$, and $\psi \in \RR^{q}$ is a column vector such that $(\phi, \psi)$ is controllable. Thus, the data-driven problems consist in determining $u_{stat}$ and $u_{dyna}$.


To ensure the existence of regulators that solve these problems, we introduce the following standard assumptions.  
\begin{assumption} \label{ass:ABstabilisability}
    The pair $(A, B)$ is stabilisable. 
\end{assumption}

\begin{assumption}[Non-resonance condition] \label{ass:nonResonance}
For all $\bar{\lambda} \in \lambda(S)$
\begin{equation*}
\rank \left(
\begin{bmatrix}
    A - \bar{\lambda} I & B \\
    C & D
\end{bmatrix}
\right) = n + p. 
\end{equation*}
\end{assumption}

It is well-kwown that under Assumptions~\ref{ass:ABstabilisability} and \ref{ass:nonResonance}, the model-based (static or dynamic) regulator problem is solvable.

For Problem 1, the design of the regulator is concerned with the cascade $\Sigma_{exo} \rightarrow \Sigma_{plant}$ and the solution of the problem is associated with the so-called regulator equations \cite[Section 1.5]{huang2004nonlinear}, which includes a Sylvester equation that characterises the invariant subspace in $\Sigma_{exo} \rightarrow \Sigma_{plant}$ and another equation that zeros the tracking error on this subspace. Such regulator equations can be solved (purely from data) by further developing the treatment presented in \cite[Section 3]{mao2025partOne}. However, we note that a solution to Problem~1 has already been presented
in \cite{trentelman2021informativity}, although from a different angle. Despite the different derivations, the final formulas obtainable with our framework are identical to those obtained in \cite{trentelman2021informativity}. Thus, we do not present this result in the main body of paper. However, for the sake of completeness and as a demonstration of the flexibility of the framework, we include the derivation from our perspective in Appendix~\ref{sec:regulationFullInfo}.

Thus, in the following we focus on Problem 2.
In the context of Problem 2, the full interconnection is characterised by the cascade $\Sigma_{exo} \rightarrow \Sigma_{plant} \rightarrow \Sigma_{imo}$. We point out that the cascade of primary interest is the latter interconnection, namely $\Sigma_{plant} \rightarrow \Sigma_{imo}$, since it is well established in the literature (see \eg{} \cite{huang2004nonlinear}) that any $u_{dyna}$ that stabilises $\Sigma_{plant} \rightarrow \Sigma_{imo}$ under $\omega \equiv 0$ readily solves this problem. However, it should be noted that the first cascade cannot be ignored in a data-driven framework. From a practical data-driven viewpoint, the stabilisation design of $\Sigma_{plant} \rightarrow \Sigma_{imo}$ still has to be performed in the presence of the unknown exogenous signal $\omega$ (\ie, under the influence of $\Sigma_{exo}$), posing theoretical challenges. In what follows, we tackle this problem.

\subsection{The Data-driven Dynamic State-feedback Problem} \label{sec:regulationRobust}
As just noted, for this problem we have to search for $u_{dyna}$ to stabilise the cascade $\Sigma_{plant} \rightarrow \Sigma_{imo}$ while using sampled data from the real cascade $\Sigma_{exo} \rightarrow \Sigma_{plant} \rightarrow \Sigma_{imo}$. This can be regarded as an extended setting for the recalled cascade stabilisation pipeline. Therefore, we consider $u_{dyna}$ in the form inherited from the forwarding method (see Section \ref{sec:data-driven-cascade-stabilisation}) as 
\begin{equation} \label{eq:regulatorRobust}
    u_{dyna}(k) := K_x x(k) + K_\zeta (\xi(k) - \Upsilon_r x(k)),
\end{equation}
where $\Upsilon_r$ is a transformation matrix to be defined shortly, and $K_x$ and $K_\zeta$ are gains to be determined. In this section, we conceptually follow steps (S1), (S2) and (S3) of the cascade stabilisation procedure recalled in Section~\ref{sec:data-driven-cascade-stabilisation}.

In order to construct the internal model $\Sigma_{imo}$, we assume that we know $\Sigma_{exo}$, \ie, the matrix $S$. This assumption is ubiquitous\footnote{In Appendix~\ref{sec:regulationFullInfo} we show that for the static problem, this assumption can be removed via a generalisation of Lemma~\ref{lemma:computePI}.} in both model-based and data-driven output regulation frameworks (sometimes, the stronger assumption of knowing $\omega$ is used). We stress that we do not assume knowledge of the initial condition $\omega(0)$, indicating that the exogenous signal $\omega(k)$ cannot be computed or measured for any $k$. 

We first mirror the action of step (S1), that is, we design the gain $K_x$ (from data) such that $\Sigma_{plant}$ is pre-stabilised, \ie, $A + B K_x$ is rendered Schur. To this end, we first establish a data-based representation of $\Sigma_{plant}$ in closed-loop with $u(k) = K_x x(k) + v(k)$, where $v(k) := K_\zeta (\xi(k) - \Upsilon_r x(k))$ is the term that is not involved at this stage but will be designed later in step (S3). Note that the dynamics of the closed-loop plant are described by
\begin{subequations} \label{eq:closed-loop-plant-model}
\begin{align} 
    x(k+1) &= (A + B K_x)x(k) + Bv(k) + E\omega(k),  \\
    e(k) &= (C + D K_x) x(k) + Dv(k) + F\omega(k), 
\end{align}  
\end{subequations}
which we denote by $\Sigma_{plant}^\prime$. To determine a data-based representation of \eqref{eq:closed-loop-plant-model} without the need to measure $\omega$
we introduce the following assumption.

\begin{assumption} \label{ass:rankXUWnew}
Let $z^* \in \RR^{\nu}$ be a vector, and let $\mathcal{W} \in \RR^{\nu \times T}$ be the Krylov matrix
\begin{equation} \label{eq:KrylovMatrixW}
    \mathcal{W} := \begin{bmatrix}
        z^* & Sz^* & S^2z^* & \dots & S^{T-1}z^*
    \end{bmatrix}.
\end{equation}
The available data matrices $X_{-}$ and $U_{-}$ are such that there exists $z^*$ such that the following rank condition is satisfied
\begin{equation} \label{eq:rankXUWNew}
    \rank \left( \begin{bmatrix}
         X_{-}  \\ U_{-} \\ \mathcal{W}
    \end{bmatrix} \right)
     = n + m + \nu. 
\end{equation}
\end{assumption}

In this assumption, $\mathcal{W}$ can be interpreted as a \textit{virtual} trajectory of $\Sigma_{exo}$ that starts from any initial condition $z^*$ that is sufficiently excitable. Then, the data-based representation of $\Sigma_{plant}^\prime$ is given as follows.

\begin{lemma} \label{theorem:closedLoopRepRegulationNew}
Consider the cascade $\Sigma_{exo} \rightarrow \Sigma_{plant}$. Suppose that Assumption \ref{ass:rankXUWnew} holds. Let $K_x \in \RR^{m \times n}$ be any matrix. Then, the following equivalent data-based representations
\begin{equation} \label{eq:closedLoopRepRegulationNew}
    A + B K_x = X_{+} G_{K_x}, \quad C+ DK_x = E_{-} G_{K_x},
\end{equation}
hold for any matrix $G_{K_x} \in \RR^{T \times n}$ that satisfies
\begin{equation} \label{eq:closedLoopRepRegulationConditionNew}
    \begin{bmatrix}
        I \\ K_x \\ \textbf{0}
    \end{bmatrix}
    =
    \begin{bmatrix}
         X_{-}  \\ U_{-} \\ \mathcal{W}
    \end{bmatrix}
    G_{K_x}.
\end{equation}
Moreover, $\Sigma_{plant}^\prime$ can be equivalently represented as
\begin{subequations}  \label{eq:data-representation-plant}
\begin{align} 
    x(k+1) &= X_{+} G_{K_x}x(k) + Bv(k) + E\omega(k),  \\
    e(k) &= E_{-} G_{K_x}x(k) + Dv(k) + F\omega(k). 
\end{align}  
\end{subequations}
\end{lemma}

\begin{proof}
We first prove the instrumental fact that there exists a matrix $\mathcal{E} \in \RR^{n \times \nu}$ such that $\mathcal{E} \mathcal{W} = \widetilde{E} \Omega_{-}$ for any $\widetilde{E}$, where $\Omega_{-}$ is the data matrix related to $\omega$. Note that this is equivalent to the statement that the row space of $\Omega_{-}$ is a subset of the row space of $\mathcal{W}$. By recognising that $\Omega_{-}$ and $\mathcal{W}$ are Krylov matrices of the same operator $S$, it follows from the properties of Krylov matrices that the statement is equivalent to that the initial condition $\omega(0)$ belongs to the column space of $\mathcal{W}$, namely
$
\omega(0) \in
\image (\mathcal{W}).
$
This condition is ensured by the full row rank of $\mathcal{W}$ implied by Assumption \ref{ass:rankXUWnew}. Thus, under Assumption \ref{ass:rankXUWnew}, the matrix $\mathcal{E}$ exists. \\
Moreover, by standard linear algebra arguments, Assumption \ref{ass:rankXUWnew} implies that there exists a matrix $G_{K_x}$ such that \eqref{eq:closedLoopRepRegulationConditionNew} holds. Then, it follows directly that
\begin{equation*}
\begin{split}
    A + B K_x &= 
    \begin{bmatrix}
        A & B & \mathcal{E}
    \end{bmatrix}
    \begin{bmatrix}
        I \\ K_x \\ \mathbf{0}
    \end{bmatrix} = 
    \begin{bmatrix}
        A & B & \mathcal{E}
    \end{bmatrix}
    \begin{bmatrix}
         X_{-}  \\ U_{-} \\ \mathcal{W}
    \end{bmatrix}
    G_{K_x} \\
    &= (A X_{-}  + B U_{-} + \mathcal{E} \mathcal{W}) G_{K_x} \\
    &= (A X_{-}  + B U_{-} + E\Omega_{-}) G_{K_x} = X_{+} G_{K_x}, 
\end{split} 
\end{equation*} 
where for the fourth equality we exploited that $\mathcal{E} \mathcal{W} = E\Omega_{-}$ while for the last equality we used \eqref{eq:sysXRegulation}. The argument for the closed-loop output matrix $C +D K_x = E_{-} G_{K_x}$ proceeds analogously,  noticing that under Assumption~\ref{ass:rankXUWnew} there exists $\mathcal{E}'$ such that $\mathcal{E}' \mathcal{W} = F\Omega_{-}$. Finally, the representation given in \eqref{eq:data-representation-plant} follows trivially from \eqref{eq:closed-loop-plant-model} by substituting \eqref{eq:closedLoopRepRegulationNew}.
\end{proof}

\begin{remark}
Lemma~\ref{theorem:closedLoopRepRegulationNew} extends the data-based representation given in \cite[Theorem 2]{de2019formulas} to the case in which the plant is under the influence of an autonomous system in the form of \eqref{eq:exosystem}. The last block row of \eqref{eq:closedLoopRepRegulationConditionNew} is an additional constraint imposed on the parametrization matrix to restrict $G_{K_x}$ to be in the kernel space of $\mathcal{W}$, hence ``hiding'' the impact of the exogenous signal $\omega$ on the representation of $A + B K_x$. This observation has also been noted in \cite[Lemma 2]{hu2025data} for a continuous-time counterpart. 
\end{remark}

Lemma \ref{theorem:closedLoopRepRegulationNew} provides an $\omega$-free representation of the closed-loop plant. This representation is instrumental in freeing the data-driven forwarding pipeline from the need of measuring $\omega$, as we show in the following results.

\begin{lemma} \label{lemma:stabilisationWithWNew}
Consider the cascade $\Sigma_{exo} \rightarrow \Sigma_{plant}$. Suppose that Assumptions \ref{ass:ABstabilisability} and \ref{ass:rankXUWnew} hold. Then any matrix $Q_{K_x} \in \RR^{T \times n}$ that satisfies
\begin{subequations} \label{eq:dataRepConditionSNew} 
  \begin{empheq}[left=\empheqlbrace]{align}
&\begin{bmatrix}
    X_{-} Q_{K_x} & X_{+} Q_{K_x} \\
    Q_{K_x}^\top X_{+}^\top &  X_{-} Q_{K_x} 
\end{bmatrix} \succ 0 \\
& \,\,\mathcal{W} Q_{K_x} = \textbf{0} \label{eq:dataRepConditionSNewb}
  \end{empheq}
\end{subequations}
is such that the feedback gain
\begin{equation} \label{eq:regulationStablizingKNew}
    K_x := U_{-} Q_{K_x} (X_{-} Q_{K_x})^{-1}
\end{equation} 
stabilises $\Sigma_{plant}$, \ie, the matrix $A + B K_x$ is rendered Schur. 
\end{lemma}

\begin{proof}
The statement follows from substituting the data-based representation of $A + B K_x$ given in Lemma \ref{theorem:closedLoopRepRegulationNew} into the proof of \cite[Theorem 3]{de2019formulas} with the additional constraint that $G_{K_x}$ has to be in the kernel space of $\mathcal{W}$. This is ensured by \eqref{eq:dataRepConditionSNewb} because, as shown in \cite[Theorem 3]{de2019formulas}, $G_{K_x} = Q_{K_x} (X_{-}Q_{K_x})^{-1}$.
\end{proof}

Having determined a pre-stabilising gain $K_x$ from solving the feasibility problem over the data-dependent LMIs and LMEs in \eqref{eq:dataRepConditionSNew}, the next step is to determine a matrix that gives a useful coordinate transformation for step (S3). To this end, we consider the solution to the Sylvester equation 
\begin{equation} \label{eq:dualSylvesterEqNew1}
    \Phi \Upsilon_r - \Upsilon_r (A + B K_x) =  - \Psi (C + D K_x).
\end{equation}
Note that this solution $\Upsilon_r$ characterises the subspace $\{(x, \xi): \xi = \Upsilon_r x\}$ that arises in 
the cascade $\Sigma_{plant}^\prime \rightarrow \Sigma_{imo}$. In the following statement, we show how to compute this $\Upsilon_r$\GS{.} 

\begin{lemma}  \label{lemma:computeUpsilonr}
Suppose that Assumptions \ref{ass:ABstabilisability} and \ref{ass:rankXUWnew} hold. Let $Q_{K_x}$ be a solution of \eqref{eq:dataRepConditionSNew}, $K_x$ be defined as in \eqref{eq:regulationStablizingKNew}, and $G_{K_x}$ be a matrix that satisfies \eqref{eq:closedLoopRepRegulationConditionNew}. Then, the following statements are equivalent\GS{.}

\begin{enumerate}
    \item[(i)] $\Upsilon_r$ solves \eqref{eq:dualSylvesterEqNew1}. 
    \item[(ii)] $\Upsilon_r$ solves 
    $\Upsilon_r X_{+} G_{K_x} = \Phi \Upsilon_r + \Psi E_{-} G_{K_x}$.
    \item[(iii)] $\Upsilon_r$ solves 
\begin{equation} \label{eq:dualSylvesterEqNew3}
\begin{aligned}
    \Upsilon_r X_{+}  Q_{K_x} & (X_{-}  Q_{K_x})^{-1} = \\
    &\Phi \Upsilon_r + \Psi E_{-}  Q_{K_x} (X_{-} Q_{K_x})^{-1}.
\end{aligned}
\end{equation}
\end{enumerate}
\end{lemma}
\begin{proof}
The equivalence (i) $\Leftrightarrow$ (ii) follows from Lemma~\ref{theorem:closedLoopRepRegulationNew}. The equivalence (ii) $\Leftrightarrow$ (iii) follows from \eqref{eq:closedLoopRepRegulationConditionNew} and \eqref{eq:regulationStablizingKNew} which imply $G_{K_x} = Q_{K_x} (X_{-} Q_{K_x})^{-1}$. 
\end{proof}

\begin{remark}
Differently from \cite[Lemma 5]{mao2025partOne} where an analogous Sylvester equation is solved via an LME, in this case $\Upsilon_r$ can be determined by solving the Sylvester equation \eqref{eq:dualSylvesterEqNew3} directly. The difference is due to the fact that the matrices of the second subsystem, in this case $\Phi$ and $\Psi$, are known here, while the second subsystem is not available in \cite[Lemma 5]{mao2025partOne}.
\end{remark}

Lemma \ref{lemma:computeUpsilonr} presents two equivalent forms of the Sylvester equation \eqref{eq:dualSylvesterEqNew1} in terms of data and parameterization matrices obtained in step (S1). Once 
\eqref{eq:dualSylvesterEqNew3} is solved, \ie, step (S2) is completed, a coordinate transformation $\zeta := \xi - \Upsilon_r x$ 
can be carried out to facilitate the final stabilisation step.

\begin{lemma}\label{OR-zetadyn}
Suppose Assumption \ref{ass:nonResonance} holds. 
Then the dynamics of the transformed internal model, denoted by $\Sigma_{imo}^{\zeta}$, is
\begin{equation} \label{eq:zetaSubsystem}
    \zeta(k+1) =  \Phi \zeta(k) + (\Psi D - \Upsilon_r B) v(k)  + (\Psi F - \Upsilon_r E) \omega(k).
\end{equation}
Moreover, the pair $(\Phi, \Psi D - \Upsilon_r B)$ is controllable.
\end{lemma}

\begin{proof}
Note that
\begin{equation} \label{eq:zetaSubsystem_proof}
\begin{aligned}
    \zeta(k+1) &=  \xi(k+1) - \Upsilon_r x(k+1) \\
    &=\! \Phi\xi(k) \!+\! \Psi \left( Cx(k) \!+\! D (K_x x(k) \!+\! v(k)) \!+\! F\omega(k)\right) \\
    & \quad - \Upsilon_r \left( Ax(k)+B (K_x x(k) + v(k)) + E\omega(k)\right) \\
    &= \Phi\xi(k) + \Psi \left( (C + D K_x) x(k) + D v(k) + F\omega(k)\right) \\
    & \quad - \Upsilon_r \left( (A+ BK_x) x(k) + B v(k) + E\omega(k)\right) \\
    &=  \Phi(\xi(k) - \Upsilon_r x(k)) + (\Psi D - \Upsilon_r B) v(k) \\
    & \quad + (\Psi F - \Upsilon_r E) \omega(k) \\
    &=  \Phi \zeta(k) + (\Psi D - \Upsilon_r B) v(k)  + (\Psi F - \Upsilon_r E) \omega(k),
\end{aligned}
\end{equation}
where in the fourth equality we exploited the Sylvester equation \eqref{eq:dualSylvesterEqNew1}. Finally, the controllability of the pair $(\Phi, \Psi D - \Upsilon_r B)$ follows from the controllability of $(\Phi, \Psi)$ along with Assumption \ref{ass:nonResonance}, see \cite[Theorem 1.a]{simpson-porco-arxiv}.
\end{proof}

Recalling that $v(k) := K_\zeta (\xi(k) - \Upsilon_r x(k))$, one can now design the gain $K_\zeta$ to stabilise the $\zeta$-dynamics. The existence of such stabilising gain is guaranteed by the controllability of $(\Phi, \Psi D - \Upsilon_r B)$. Then, by noticing that the $\zeta$-dynamics $\eqref{eq:zetaSubsystem}$ is associated to the cascade $\Sigma_{exo} \rightarrow \Sigma_{imo}^{\zeta}$ with input $v$, the stabilisation of $\Sigma_{imo}^{\zeta}$ can be achieved by applying another iteration of the methodology established in Lemma \ref{lemma:stabilisationWithWNew} for $\Sigma_{exo} \rightarrow \Sigma_{plant}$. The approach is formalised as follows.

\begin{theorem} \label{theorem:SolutionRobustRegulationProblem}
Consider the full cascade $\Sigma_{exo} \rightarrow \Sigma_{plant} \rightarrow \Sigma_{imo}$ and suppose that Assumptions~\ref{ass:ABstabilisability}, \ref{ass:nonResonance}, and \ref{ass:rankXUWnew} hold. Let $K_x$ be constructed as in \eqref{eq:regulationStablizingKNew} with $Q_{K_x}$ any solution of \eqref{eq:dataRepConditionSNew}. Let $\Upsilon_r$,  the solution of equation \eqref{eq:dualSylvesterEqNew3}, be such that 
\begin{equation} \label{eq:rankZVWNew}
    \rank \left( \begin{bmatrix}
         Z_{-}  \\ V_{-} \\ \mathcal{W}
    \end{bmatrix} \right)
     = pq + m + \nu,
\end{equation}
with $Z_{-} := \Xi_{-} - \Upsilon_r X_{-}$ and $V_{-} := U_{-} - K_x X_{-}$, holds. Then, the regulator \eqref{eq:regulatorRobust},
with 
\begin{equation} \label{eq:post-stabilising-Kzeta}
    K_\zeta := V_{-} Q_{K_\zeta} (Z_{-} Q_{K_\zeta})^{-1}, 
\end{equation}
where $Q_{K_\zeta} \in \RR^{T \times pq}$ is any matrix that satisfies
\begin{subequations} \label{eq:dataRepConditionSNewZeta} 
  \begin{empheq}[left=\empheqlbrace]{align}
&\begin{bmatrix}
    Z_{-} Q_{K_\zeta} & Z_{+} Q_{K_\zeta} \\
    Q_{K_\zeta}^\top Z_{+}^\top &  Z_{-} Q_{K_\zeta} 
\end{bmatrix} \succ 0 \\
& \mathcal{W} Q_{K_\zeta} = \textbf{0}
  \end{empheq}
\end{subequations}
with $Z_{+} := \Xi_{+} - \Upsilon_r X_{+}$, 
is such that the cascade $\Sigma_{plant} \rightarrow \Sigma_{imo}$ is asymptotically stable.
\end{theorem}
\begin{proof}
Note that the constructed data matrices $Z_{-/+}$ and $V_{-}$ are the associated state and input data of $\Sigma_{imo}^{\zeta}$. Then, the result that the matrix $\Phi + (\Psi D - \Upsilon_r B) K_\zeta$ is rendered Schur follows directly from Lemma \ref{lemma:stabilisationWithWNew} by replacing the associated data of $\Sigma_{plant}$ with the associated data of $\Sigma_{imo}^{\zeta}$. \\
Having stabilised both $\Sigma_{plant}$ and $\Sigma_{imo}^\zeta$, \ie, the matrices $A + B K_x$ and $\Phi + (\Psi D - \Upsilon_r B) K_\zeta$ are Schur, we note that the dynamics of the cascade $\Sigma_{plant} \rightarrow \Sigma_{imo}$ in closed-loop with the regulator $u_{dyna}$ are described by
\begin{align*} 
    \begin{bmatrix}
        x(k+1)\\\zeta(k+1)
    \end{bmatrix} &= \begin{bmatrix}
        A + B K_x & B K_\zeta \\ 
        \mathbf{0} & \Phi + (\Psi D - \Upsilon_r B) K_\zeta
    \end{bmatrix}\begin{bmatrix}
        x(k)\\\zeta(k)
    \end{bmatrix} \\
    &\quad \; +
    \begin{bmatrix}
        E \\
        \Psi F - \Upsilon_r E
    \end{bmatrix}
    \omega(k)
\end{align*}
This is a triangular system and so is asymptotically stable. In summary, the synthesised $u_{dyna}$ solves Problem 2. 
\end{proof}

\begin{remark}
Since the problem has been recast into a data-driven cascade stabilisation problem, if there are measurement and/or process noise in the cascade, then the regulator can be synthesised in an analogous way to the controller design proposed in Part I, see \cite[Section IV.C]{mao2025partOne}, to ensure that the output lies in a bounded ball (the radius of which depends on the norms of the uncertainties) centred on the reference motion. The results in \cite[Section IV.C]{mao2025partOne} apply readily and thus they are omitted to avoid repetitions.
\end{remark}

While so far the results have been presented for the state-feedback case for convenience, the need to measure the state rather than the output can be trivially removed. This can be achieved by utilising the left difference operator representation, as already recalled in Section~\ref{sec:reductionIOdata} for the model order reduction problem. Consider 
the left difference operator representation of the plant \eqref{eq:sysXRegulation}-\eqref{eq:sysERegulation}, given as
\begin{equation} \label{eq:leftDiffRepReg}
\begin{split}
    e(k) + \sum_{i=1}^n A_i e(k-i) = \sum_{i=0}^n B_i u(k-i) + \sum_{i=0}^n E_i \omega(k-i),
\end{split}
\end{equation}
where $A_i \in \RR^{p \times p}$, $B_i \in \RR^{p \times m}$, and $E_i \in \RR^{p \times \nu}$ are coefficient matrices\footnote{There exists a one-to-one relationship between these matrices and the plant matrices in $\Sigma_{plant}$.}. Consider the augmented state
\begin{equation} \label{eq:xTildeReg}
\begin{gathered}
       \widetilde{x}(k) :=\col \bigl( e(k-n), e(k-n+1), \ldots, e(k-1) \\
u(k-n), u(k-n+1), \ldots, u(k-1) \bigr) \in \RR^{n(m + p)},
\end{gathered}
\end{equation}
and note that the last term in \eqref{eq:leftDiffRepReg} can be rewritten in terms of only the ``present'' $\omega$ as
$
    \left(\sum_{i=1}^n E_i S^{-i}\right) \omega(k).
$
Then, analogously to Section~\ref{sec:reductionIOdata}, it follows from \eqref{eq:leftDiffRepReg} that there exists matrices $\widetilde{A}$, $\widetilde{B}$, $\widetilde{C}$, $\widetilde{D}$, $\widetilde{E}$, and $\widetilde{F}$ (of proper dimensions) such that the state-space realisation 
\begin{subequations}  \label{eq:stateSpaceBigReg}
\begin{align} 
    \widetilde{x}(k+1) &= \widetilde{A}\widetilde{x}(k) + \widetilde{B}u(k) + \widetilde{E}\omega(k),  \\
    e(k) &= \widetilde{C}\widetilde{x}(k) + \widetilde{D}u(k) + \widetilde{F}\omega(k),
\end{align}  
\end{subequations} 
equivalently captures the input-output behaviour of the plant. As such, with the additional assumption that the pair $ (A, C) $ is detectable (to guarantee that the convergence of $e$ to zero implies the same for $x$), all instances of the original state matrices $X_{-}$  and \( X_{+} \) in the results presented earlier in this section can be replaced by the extended state matrices $\widetilde{X}_{-}$ and $\widetilde{X}_{+}$, which rely solely on input-output measurements. This substitution preserves the stabilisation result of the cascade $\Sigma_{plant} \rightarrow \Sigma_{imo}$, thereby ensuring that Problem~2 is solved. 

In summary, we summarise in Algorithm \ref{alg:dataDrivenRobustRegulation} the procedure to synthesise the regulator solving Problem~2
from an \textit{input-output} perspective.   
\begin{algorithm}
  \caption{Data-Driven Dynamic Error-Feedback Output Regulation}
\label{alg:dataDrivenRobustRegulation}
  \begin{algorithmic}[1]
    \STATE \textbf{Input:} Collected samples of $u$ and $e$ from $\Sigma_{plant} \rightarrow \Sigma_{imo}$, $S$
    \STATE \textbf{\# \, Assemble data matrices}
    \STATE \hspace{1em} Construct matrices $\Phi$ and $\Psi$ using \eqref{eq:p-internal-model}
    \STATE \hspace{1em} Compute the trajectory of $\xi$ with any $\xi(0)$ from \eqref{eq:internalModel}
    \STATE \hspace{1em} Construct data matrices $(U_-, \widetilde{X}_-, \widetilde{X}_+, \Xi_-, \Xi_+)$
    \STATE \hspace{1em} Let $X_- := \widetilde{X}_-$ and $X_+ := \widetilde{X}_+$
    \STATE \hspace{1em} Select $z^*$ such that $\mathcal{W}$ given by \eqref{eq:KrylovMatrixW} satisfies \eqref{eq:rankZVWNew}
    \STATE \textbf{\# \, Pre-stabilise the $x$-dynamics}
    \STATE \hspace{1em} Solve the LMI \eqref{eq:dataRepConditionSNew} to obtain $Q_{K_x}$
    \STATE \hspace{1em} Construct the pre-stabilising gain $K_x$ using \eqref{eq:regulationStablizingKNew}
    \STATE \textbf{\# \, Find the coordinate transformation}
    \STATE \hspace{1em} Solve the Sylvester equation \eqref{eq:dualSylvesterEqNew3} to obtain $\Upsilon_r$
    \STATE \textbf{\# \, Stabilise the $\zeta$-dynamics}
    \STATE\hspace{1em} Construct the transformed system data:\\
      \hspace{1em}$\begin{array}{l}
            V_{-} := U_{-} - K_x X_{-}, \,\,\,
            Z_{-/+} := \Xi_{-/+} - \Upsilon_r X_{-/+}\end{array}
      $
    \STATE \hspace{1em} Solve the LMI \eqref{eq:dataRepConditionSNewZeta} to obtain $Q_{K_\zeta}$
    \STATE \hspace{1em} Construct the stabilising gain $K_\zeta$ using \eqref{eq:post-stabilising-Kzeta}
    \STATE \textbf{Output:} Regulator $u_{dyna}$ in \eqref{eq:regulatorRobust} with $x$ replaced by $\widetilde{x}$
  \end{algorithmic}
\end{algorithm}






\subsection{Numerical Example} \label{sec:exampleRegulation}
Consider the linearised model of an unstable batch reactor process, studied in \eg{} \cite{rosenbrock1974cacs, walsh2001scheduling, de2019formulas}. The model is discretised using a sampling time $t_s = 0.05$~s. The resulting discrete-time model is given by
\begin{align*}
& [A|B] =  \\
&
\left[
\begin{array}{rrrr|rr}
  1.080 & -0.005 &  0.290 & -0.237 &  0.001 & -0.024 \\
 -0.027 &  0.810 & -0.003 &  0.032 &  0.257 &  0.000 \\
  0.045 &  0.189 &  0.732 &  0.235 &  0.084 & -0.135 \\
  0.001 &  0.189 &  0.055 &  0.912 &  0.084 & -0.005
\end{array}
\right],
\end{align*}
with the state vector $x(k) = [x_1(k), x_2(k), x_3(k), x_4(k)]^\top \in \RR^{4}$, with $x_1(k)$ the reactor bulk temperature, $x_2(k)$ the coolant-jacket temperature, $x_3(k)$ the (scaled) concentration of reactant, and $x_4(k)$ the heat-removal wall temperature lag. One measured output (of typical interest) is $y(k) = x_1(k) + x_3(k) - x_4(k)$, which corresponds to the reactor-thermocouple reading. The problem of interest is to regulate the output $y$ towards a sinusoidal reference signal $r$, while rejecting the influence of a high-frequency voltage fluctuation on the external coolant power supply (\ie, a disturbance $d$ on $x_2(k)$). The above consideration gives rise to the following matrices
\begin{align*}
[C|D] =
\left[
\begin{array}{rrrr|rr}
  1 &  0 &  1 & -1 &  0 & 0 \\
\end{array}
\right],
\end{align*}
\begin{equation*}
\begin{bmatrix}
    E \\ \hline
    F
\end{bmatrix}
=
\begin{bmatrix}
  0 &  0 &  0 &  0 \\
  0 &  0 &  -1 &  1 \\ 
  0 &  0 &  0 &  0 \\
  0 &  0 &  0 &  0 \\ \hline
  -1 &  0 &  0 &  0 \\
\end{bmatrix},
\end{equation*}
while the exosystem is characterised by the matrix
$$
\!S \!=\!\!  \left[\begin{array}{rcrc}\!\!\!
        \cos(\pi t_s) & \sin(\pi t_s) & 0 & 0 \\
        -\sin(\pi t_s) & \cos(\pi t_s) & 0 & 0 \\
        0 & 0 & \cos(7.5\pi t_s) & \sin(7.5\pi t_s) \\
        0 & 0 & -\sin(7.5\pi t_s) & \cos(7.5\pi t_s)        
    \end{array}\!\!\!\right]\!,
$$
where the first diagonal block corresponds to the reference $r$, characterised by a frequency of $0.5$~Hz, whereas the second diagonal block corresponds to the disturbance $d$, characterised by a frequency of $3.75$~Hz.

We test the state-feedback method established in Section \ref{sec:regulationRobust}, as well as its error-feedback version, namely Algorithm \ref{alg:dataDrivenRobustRegulation}. A minimal number (in terms of satisfying the various rank conditions) of data points is collected for both approaches. The LMIs and LMEs are solved using MATLAB CVX \cite{cvx} with SeDuMi \cite{sturm1999using}. Fig. \ref{fig:regulation_tracking} (top) shows the time histories of the reference signal (solid/black) and and of the system outputs under the controller resulting from both the state-feedback (dashed/blue) and the error-feedback approach (dash-dotted/green), verifying that the system output for both approaches tracks the reference asymptotically, despite the high-frequency disturbance. This is confirmed by the respective tracking errors $e$, shown in Fig. \ref{fig:regulation_tracking} (bottom).

\begin{figure}
    \centering
    \includegraphics[width=0.98\columnwidth]{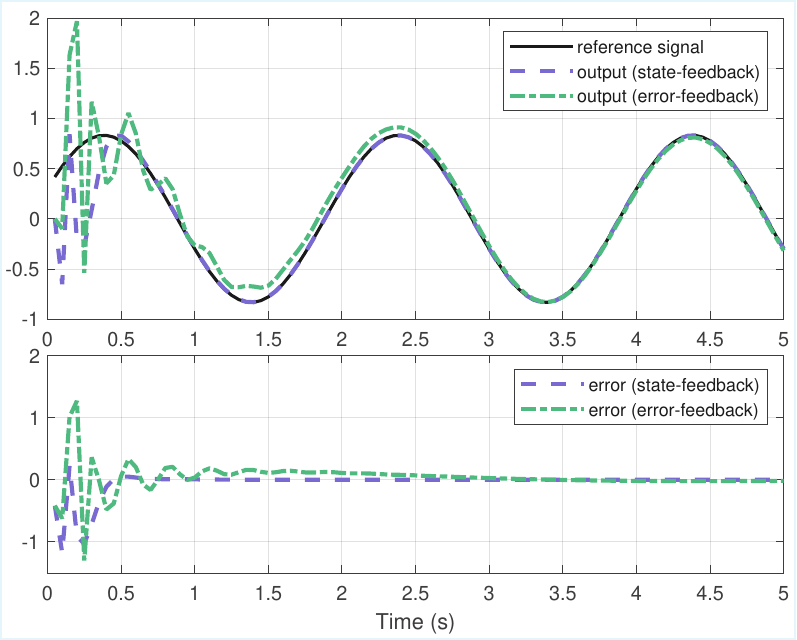}
    \caption{(Top) Time histories of the reference signal $r(k)$ (solid/black line), the output signal for the state-feedback approach (dashed/blue line), and the output signal for the error-feedback approach (dash-dotted/green line). (Bottom) Corresponding tracking errors.}
    \label{fig:regulation_tracking}
\end{figure}


\color{black}

\section{Conclusion} \label{sec:conclusion}
In this article, a new interconnection-based perspective has been introduced to address the direct data-driven design of two problems, model order reduction and output regulation, spanning both modelling and control domains. The proposed solutions were developed under the assumption of unknown system dynamics and built on a unified methodology presented in Part I of this article. Together, these results illustrate the scope and versatility of the established framework for addressing a class of control-theoretic problems interpretable through cascade interconnections.

While the article has scratched the surface with only three representative problems due to space limitations, we believe that an extended range of data-driven problems can be formulated, revisited or novelly addressed within the proposed framework. Potential problems include observer design \cite{luenberger1964observing}, eigenstructure assignment \cite{shafai1988algorithm}, disturbance decoupling \cite{syrmos2002disturbance}, hierarchical control \cite{girard2009hierarchical}. 

Another key direction for future research lies in extending the interconnection-based framework to nonlinear systems, wherein Sylvester equations are replaced by their nonlinear counterparts: invariance partial differential equations. A successful generalisation in this direction could pave the way for a systematic, data-driven design methodology applicable to a wide range of nonlinear control problems.

\section*{Acknowledgment}
J. Mao would like to thank Zirui Niu for helpful discussions regarding the output regulation problem. G. Scarciotti would like to thank the members of the ``Sylvester Kings'' group, Daniele Astolfi and John Simpson-Porco, for daily illuminating discussions over several years regarding control theoretic implications of the Sylvester equation. 

\section*{References}
\bibliographystyle{IEEEtran}
\bibliography{ref}

\appendix
\subsection{The Extended Configuration: Both Subsystems are with Unknown Dynamics} 
\label{sec:extendedConfig}
In Lemma~\ref{lemma:computePI}, we provided a data-driven solution of the Sylvester equation for the scenario in which $\mathbf{\Sigma_1}$ is known and $\mathbf{\Sigma_2}$ is unknown. This scenario applies to all the results presented in Parts I and II. In this section we extend that lemma to the case in which also $\mathbf{\Sigma_1}$ is unknown. This extended case is relevant for a result presented in Appendix~\ref{sec:regulationFullInfo}.

\begin{theorem} \label{corollary:extendedConfig}
Consider the cascade $\mathbf{\Sigma_1} \rightarrow \mathbf{\Sigma_2}$, where both $\mathbf{\Sigma_1}$ and $\mathbf{\Sigma_2}$ are unknown, and the associated Sylvester equation \eqref{eq:SylEqGeneric}. Suppose that the rank condition \eqref{eq:rankXU} holds 
and that $\lambda(\mathbf{A_1}) \cap \lambda(\mathbf{A_2}) = \emptyset$. Let $u_1 \equiv 0$ and suppose that $    
    \rank \left(
        X_{1, -}  
    \right)
     = n_1,
$
Then any matrix $G_\Theta \in \RR^{T \times n_1}$ that satisfies
\begin{subequations} \label{eq:dataRepPiGeneralised} 
  \begin{empheq}[left=\empheqlbrace]{align}
    X_{2, +}  G_\Theta  X_{1, -} &= X_{2, -} G_\Theta  X_{1, +}   \label{eq:dataRepPiGeneralised1} \\
    U_{2, -}  G_\Theta  X_{1, -} &=  Y_{1, -}    \label{eq:dataRepPiGeneralised2}
  \end{empheq}
\end{subequations}
is such that
\begin{equation} \label{eq:parameterizedPi2} 
    \mathbf{\Theta} := X_{2, -} G_\Theta
\end{equation}
is the solution of \eqref{eq:SylEqGeneric}.
Conversely, the solution of \eqref{eq:SylEqGeneric} can be written as in \eqref{eq:parameterizedPi2}, with $G_\Theta$ solution of \eqref{eq:dataRepPiGeneralised}.  
\end{theorem}

\begin{proof}
We first prove that 
\eqref{eq:dataRepPi} $\Rightarrow$ \eqref{eq:dataRepPiGeneralised}. To this end, we post-multiply \eqref{eq:dataRepPi} with $X_{1, -}$,  yielding
\begin{align*}
    X_{2, +}  G_\Theta X_{1, -} &= X_{2, -} G_\Theta \mathbf{A_1} X_{1, -},   \\
    U_{2, -}  G_\Theta X_{1, -} &= \mathbf{C_1} X_{1, -}. 
\end{align*}
These two equations simplify to \eqref{eq:dataRepPiGeneralised} by recognising that $\mathbf{A_1} X_{1, -} = X_{1, +}$ under $u_1 \equiv 0$ and $\mathbf{C_1} X_{1, -} = Y_{1, -}$. To prove that \eqref{eq:dataRepPiGeneralised} $\Rightarrow$ \eqref{eq:dataRepPi}, recall that since
$    
    \rank \left(
        X_{1, -}  
    \right)
     = n_1
$, there exists a right inverse $X_{1, -}^\dagger$. We post-multiply \eqref{eq:dataRepPiGeneralised} by $X_{1, -}^\dagger$, yielding the claim by noticing that 
$X_{1, -} X_{1, -}^\dagger = I$, $X_{1, +} X_{1, -}^\dagger = \mathbf{A_1}$, and $Y_{1, -} X_{1, -}^\dagger = \mathbf{C_1}$. Finally, based on the equivalence between \eqref{eq:dataRepPiGeneralised} and \eqref{eq:dataRepPi}, the statement follows directly from Lemma \ref{lemma:computePI}. 
\end{proof}

\subsection{The Data-driven Static State-feedback Problem} \label{sec:regulationFullInfo}
In this section we solve Problem 1 in Section \ref{sec:output-regulation-formulation}. Consider the regulator $u_{stat}$ given by 
\begin{equation} \label{eq:regulatorFullInfo}
    u_{stat}(k) := K_x x(k) + K_\omega \omega(k)\tm{,} 
\end{equation}
with $K_x$ and $K_\omega$ constant matrices of appropriate dimensions. 
Note that the cascade $\Sigma_{exo} \rightarrow \Sigma_{plant}$ in closed-loop with $u_{stat}$ is described by
\begin{equation} \label{eq:sysRegulationClosedLoop}
\begin{split} 
    \omega(k+1) &= S \omega(k),\\
    x(k+1) &= (A + B K_x)x(k) + (B K_\omega + E)\omega(k), \\
    e(k) &= (C + D K_x)x(k)+ (D K_\omega + F) \omega(k). 
\end{split}
\end{equation}  
Therefore, objective (O1) can be translated into the task of determining a gain $K_x$ such that matrix $A + BK_x$ is rendered Schur. This is exactly the problem we solved in Lemma \ref{lemma:stabilisationWithWNew}, hence we do not repeat its derivation here. The remaining task is to achieve objective~(O2) by suitably designing the gain $K_\omega$ from data. To this end, by pivoting the focus back to the closed-loop system \eqref{eq:sysRegulationClosedLoop}, we see that there exists an invariant subspace $\mathcal{M}_r := \{(x, \omega) \in \RR^{n + \nu}: x = \Pi_r \omega\}$, with $\Pi_r$ the solution of the Sylvester equation
\begin{equation} \label{eq:regulatorEq1Orginal}
    (A + B K_x) \Pi_r - \Pi_r S=  - (B K_\omega + E).
\end{equation}
In particular, $\mathcal{M}_r$ is attractive, \ie, $\lim_{k \rightarrow \infty} x(k) - \Pi_r \omega(k) = 0$, since $(A + B K_x)$ is rendered Schur in the process of achieving objective (O1). Then, objective~(O2) can be translated into two equations, namely the Sylvester equation $\eqref{eq:regulatorEq1Orginal}$, which gives the (attractive) invariant subspace $\mathcal{M}_r$, and an additional linear matrix equation
\begin{equation} \label{eq:regulatorEq2Orginal}
    \mathbf{0} = (C+ D K_x) \Pi_r + D K_\omega + F,
\end{equation}
which ensures that the tracking error $e$ is zero on $\mathcal{M}_r$. By defining a change of variable 
\begin{equation} \label{eq:GammaChange}
    \Gamma_r := K_x \Pi_r + K_\omega,
\end{equation}
\eqref{eq:regulatorEq1Orginal}-\eqref{eq:regulatorEq2Orginal} lead to the so-called \textit{regulator equations}
\begin{subequations}  \label{eq:regulationEqs}
\begin{align} 
\Pi_rS= A \Pi_r + B \Gamma_r + E, \label{eq:regulationEqs1} \\
\mathbf{0} = C\Pi_r + D\Gamma_r + F, \label{eq:regulationEqs2}
\end{align}
\end{subequations}
in the unknowns $(\Pi_r, \Gamma_r)$. Note that the solvability of \eqref{eq:regulationEqs} for any matrix $E$ and $F$ is sufficiently and necessarily guaranteed by the non-resonance condition in Assumption \ref{ass:nonResonance} \cite{huang2004nonlinear}. In what follows, we demonstrate how these regulator equations can be solved without knowing $\Sigma_{plant}$. 
We begin by introducing the rank condition on the data matrices.

\begin{assumption} \label{ass:rankXUW}
The available data matrices are such that the following rank condition is satisfied
\begin{equation} \label{eq:rankXUW}
    \rank \left( \begin{bmatrix}
         X_{-}  \\ U_{-} \\ \Omega_{-}
    \end{bmatrix} \right)
     = n + m + \nu.
\end{equation}
\end{assumption}

\begin{theorem}
\label{thm:staticinformativity}
Consider the cascade system $\Sigma_{exo} \rightarrow \Sigma_{plant}$ and the regulator equations \eqref{eq:regulationEqs}. Suppose that Assumptions \ref{ass:ABstabilisability}, \ref{ass:nonResonance} and \ref{ass:rankXUW} hold. Then any matrix $G_r \in \RR^{T \times \nu}$ satisfying
\begin{subequations} \label{eq:dataRepConditionR} 
  \begin{empheq}[left=\empheqlbrace]{align}
X_{+} G_r &= X_{-} G_r S  \label{eq:dataRepConditionR1}  \\
E_{-} G_r &= \textbf{0} \label{eq:dataRepConditionR2} \\
\Omega_{-} G_r &= I  \label{eq:dataRepConditionR3}
  \end{empheq}
\end{subequations}
is such that 
\begin{equation} \label{eq:regulation-equation-parameterisation}
    \Pir := X_{-}G_r, \quad \Gar := U_{-} G_r
\end{equation}
solve \eqref{eq:regulationEqs}. 
Conversely, any solution of \eqref{eq:regulationEqs} can be written as in \eqref{eq:regulation-equation-parameterisation}
with $G_r$ solution of \eqref{eq:dataRepConditionR}. 

\end{theorem}
\begin{proof}
\textit{(Necessity)} Suppose that $\Pir$ and $\Gar$ are a solution of the regulator equations \eqref{eq:regulationEqs}, which can be rewritten in matrix form as
\begin{equation} \label{eq:RegEqMat}
    \begin{bmatrix}
        A & B & E \\
        C & D & F 
    \end{bmatrix}
    \begin{bmatrix}
        \Pir \\ \Gar \\ I
    \end{bmatrix}
    = 
    \begin{bmatrix}
        \Pir S \\ \mathbf{0}
    \end{bmatrix}.
\end{equation}
By the rank condition \eqref{eq:rankXUW}, it follows from Rouché–Capelli theorem that there exists a matrix \( G_r \) such that
\begin{equation} \label{eq:parametrizationPiGaReg}
     \begin{bmatrix}
        \Pir \\ \Gar \\ I
    \end{bmatrix}
    = 
    \begin{bmatrix}
         X_{-}  \\ U_{-} \\ \Omega_{-}
    \end{bmatrix}
    G_r.
\end{equation}
Substituting \eqref{eq:parametrizationPiGaReg} into the matrix form \eqref{eq:RegEqMat} yields that
\begin{equation*}
\begin{split}
   \begin{bmatrix}
        A & B & E \\
        C & D & F 
    \end{bmatrix}
\begin{bmatrix}
         X_{-}  \\ U_{-} \\ \Omega_{-}
    \end{bmatrix}
    G_r
    = 
    \begin{bmatrix}
        X_{-} G_r S \\ \mathbf{0}
    \end{bmatrix},
\end{split}
\end{equation*}
which, by further using the plant equations \eqref{eq:sysXRegulation} and \eqref{eq:sysERegulation}, gives that
\begin{equation*}
\begin{split}
\begin{bmatrix}
         X_{+}  \\ E_{-}
    \end{bmatrix}
    G_r
    = 
    \begin{bmatrix}
        X_{-} G_r S \\ \mathbf{0}
    \end{bmatrix},
\end{split}
\end{equation*}
and so it verifies the conditions~\eqref{eq:dataRepConditionR1} and \eqref{eq:dataRepConditionR2}. In addition, the last block row of \eqref{eq:parametrizationPiGaReg} yields \eqref{eq:dataRepConditionR3} directly. 

\textit{(Sufficiency)} Suppose that $G_r$ is a solution of \eqref{eq:dataRepConditionR}. We first show that $X_{-} G_r$ and $U_{-} G_r$ solves \eqref{eq:regulationEqs1}.  It follows that 
\begin{align*}
    A (X_{-} G_r) & + B (U_{-} G_r) + E \\
    \overset{\eqref{eq:dataRepConditionR3}}{=}& A (X_{-} G_r) + B (U_{-} G_r) + E (\Omega_{-} G_r) \\
    = \,\,& (A X_{-} + B U_{-} + E \Omega_{-}) G_r \\
    \overset{\eqref{eq:sysXRegulation}}{=}& X_{+}  G_r
    \overset{\eqref{eq:dataRepConditionR1}}{=} (X_{-} G_r) S,
\end{align*}
verifying that $(X_{-} G_r, U_{-} G_r)$ constitutes 
a solution of \eqref{eq:regulationEqs1}. Next, we prove that 
$(X_{-} G_r, U_{-} G_r)$ also solves \eqref{eq:regulationEqs2}. To this end, we observe that 
\begin{align*}
    C (X_{-} G_r) & + D (U_{-} G_r) + F \\
    \overset{\eqref{eq:dataRepConditionR3}}{=}& C (X_{-} G_r) + D (U_{-} G_r) + F (\Omega_{-} G_r) \\
    = \,\,& (C X_{-} + D U_{-} + F \Omega_{-}) G_r \\
    \overset{\eqref{eq:sysERegulation}}{=}& E_{-}  G_r
    \overset{\eqref{eq:dataRepConditionR2}}{=} \mathbf{0},
\end{align*}
This verifies that $(X_{-} G_r, U_{-} G_r)$  also satisfies  
\eqref{eq:regulationEqs2}, which, together with the previous arguments for \eqref{eq:regulationEqs1}, proves that $(X_{-}G_r, U_{-} G_r)$ solves the regulator equations \eqref{eq:regulationEqs}. 
\end{proof}

As already noted, a solution to Problem~1 has already been presented
in \cite{trentelman2021informativity}, although from a different angle. Despite the different derivations,   the formula \eqref{eq:dataRepConditionR} and \cite[(21)]{trentelman2021informativity} are identical. We have included this result for the sake of completeness and as a demonstration of the flexibility of framework (which can solve both the static and the dynamic output regulation problem).

Next, inspired by the setting studied in Appendix~\ref{sec:extendedConfig}, where the first subsystem can be also unknown,
we show that Problem 1 can be solved even when $S$ is not known.

\begin{corollary}
Without assuming knowledge of $S$, Theorem~\ref{thm:staticinformativity} holds when replacing \eqref{eq:dataRepConditionR} with
\begin{subequations} \label{eq:dataRepConditionR-noS} 
  \begin{empheq}[left=\empheqlbrace]{align}
X_{+} G_r \Omega_{-} &= X_{-} G_r \Omega_{+}  \label{eq:dataRepConditionR-noS1}  \\
E_{-} G_r &= \textbf{0} \label{eq:dataRepConditionR-noS2} \\
\Omega_{-} G_r &= I  \label{eq:dataRepConditionR-noS3}
  \end{empheq}
\end{subequations}
\end{corollary}
\begin{proof}
Noting that the matrix $\Omega_{-}$ has full row rank by Assumption \ref{ass:rankXUW}, then there exists a right inverse $\Omega_{-}^\dagger$. Then, post-multiplying \eqref{eq:dataRepConditionR-noS1} by $\Omega_{-}^\dagger$ yields \eqref{eq:dataRepConditionR1}. Conversely, post-multiplying \eqref{eq:dataRepConditionR1} with $\Omega_{-}$, yields \eqref{eq:dataRepConditionR-noS1}. 
\end{proof}

The equations \eqref{eq:dataRepConditionR} or \eqref{eq:dataRepConditionR-noS} serve as a data-dependent reformulation of the regulation equations \eqref{eq:regulationEqs}, allowing for computing the solution $(\Pi_r, \Gamma_r)$ directly from data through solving the feasibility problem over LMEs. Once matrix $\Pi_r$ and $\Gamma_r$ are obtained, the selection
$
K_\omega := \Gamma_r - K_x \Pi_r,
$
with $K_x$ any stabilising gain obtained from Lemma~\ref{lemma:stabilisationWithWNew}, solves Problem~1.

\begin{IEEEbiography}[{\includegraphics[width=1in,height=1.25in,clip,keepaspectratio]{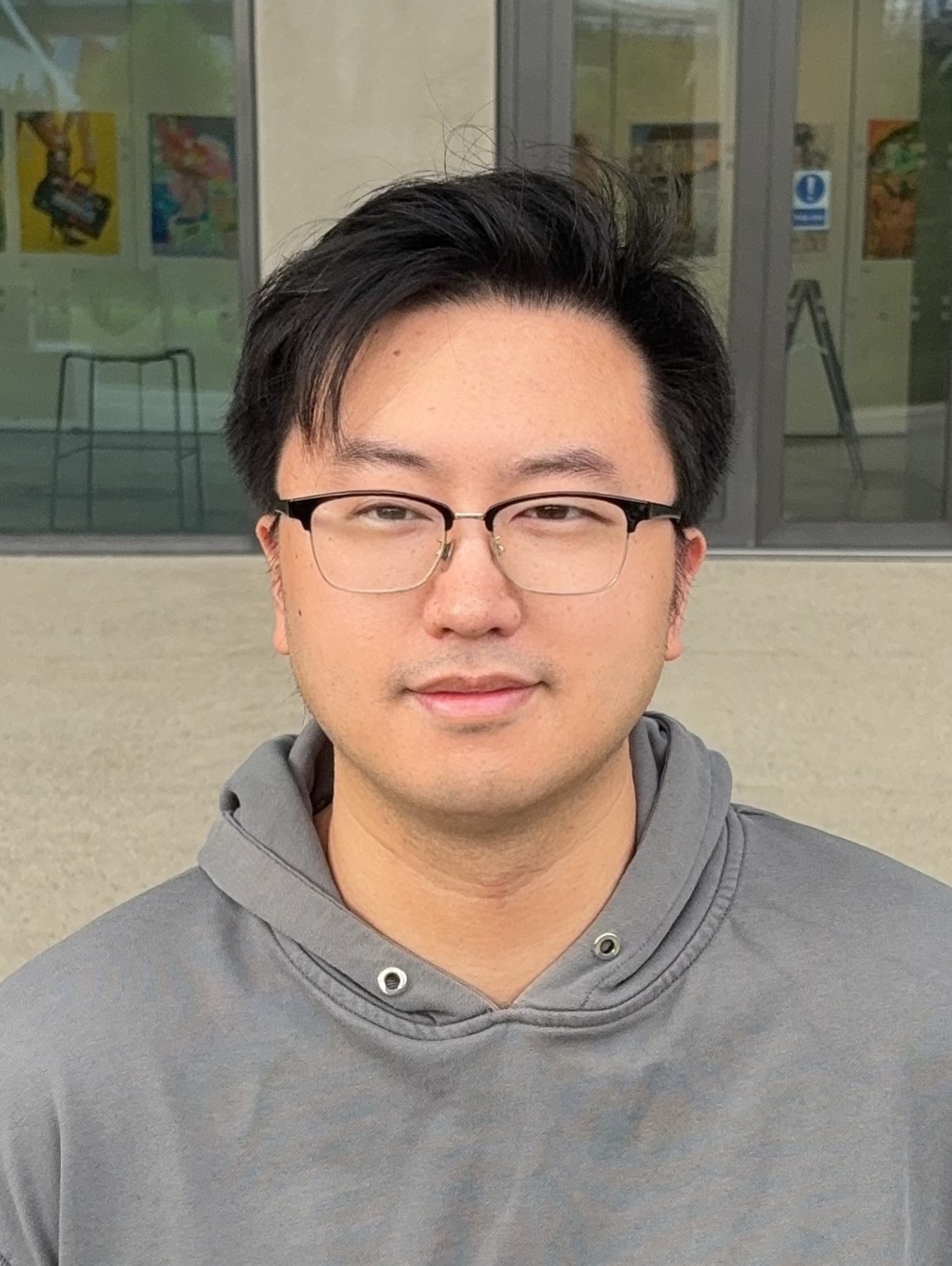}}]{Junyu Mao} (Student Member, IEEE) received his B.Eng. degree in Electrical and Electronic Engineering from the University of Liverpool, Liverpool, U.K., in 2018, and two M.Sc. degrees---in Control Systems from Imperial College London, London, U.K., in 2019, and in Data Science and Machine Learning from University College London, London, U.K., in 2020. He is currently pursuing the Ph.D. degree in the Control and Power Group at Imperial College London. In December 2024, he was a visiting Ph.D. student at the French National Centre for Scientific Research (CNRS). His research interests include the theoretical foundations of model order reduction and data-driven control for large‑scale dynamical systems.
\end{IEEEbiography}

\begin{IEEEbiography}[{\includegraphics[width=1in,height=1.25in,clip,keepaspectratio]{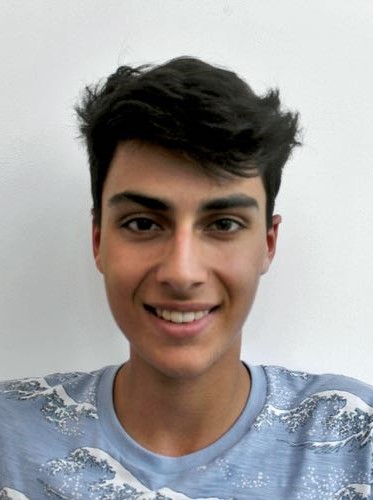}}]{Emyr Williams} is a Ph.D Candidate at the MAC-X Lab, Imperial College London. He received his M.Eng. degree in Aeronautical Engineering from Imperial College London in 2022. Following two years in industry as a software engineer, he joined the Control and Power research group at Imperial College London in 2024 to pursue his Ph.D. His research specialises on the application of data-driven and nonlinear control methodologies in the offshore renewable sector.
\end{IEEEbiography}

\begin{IEEEbiography}[{\includegraphics[width=1in,height=1.25in,clip,keepaspectratio]{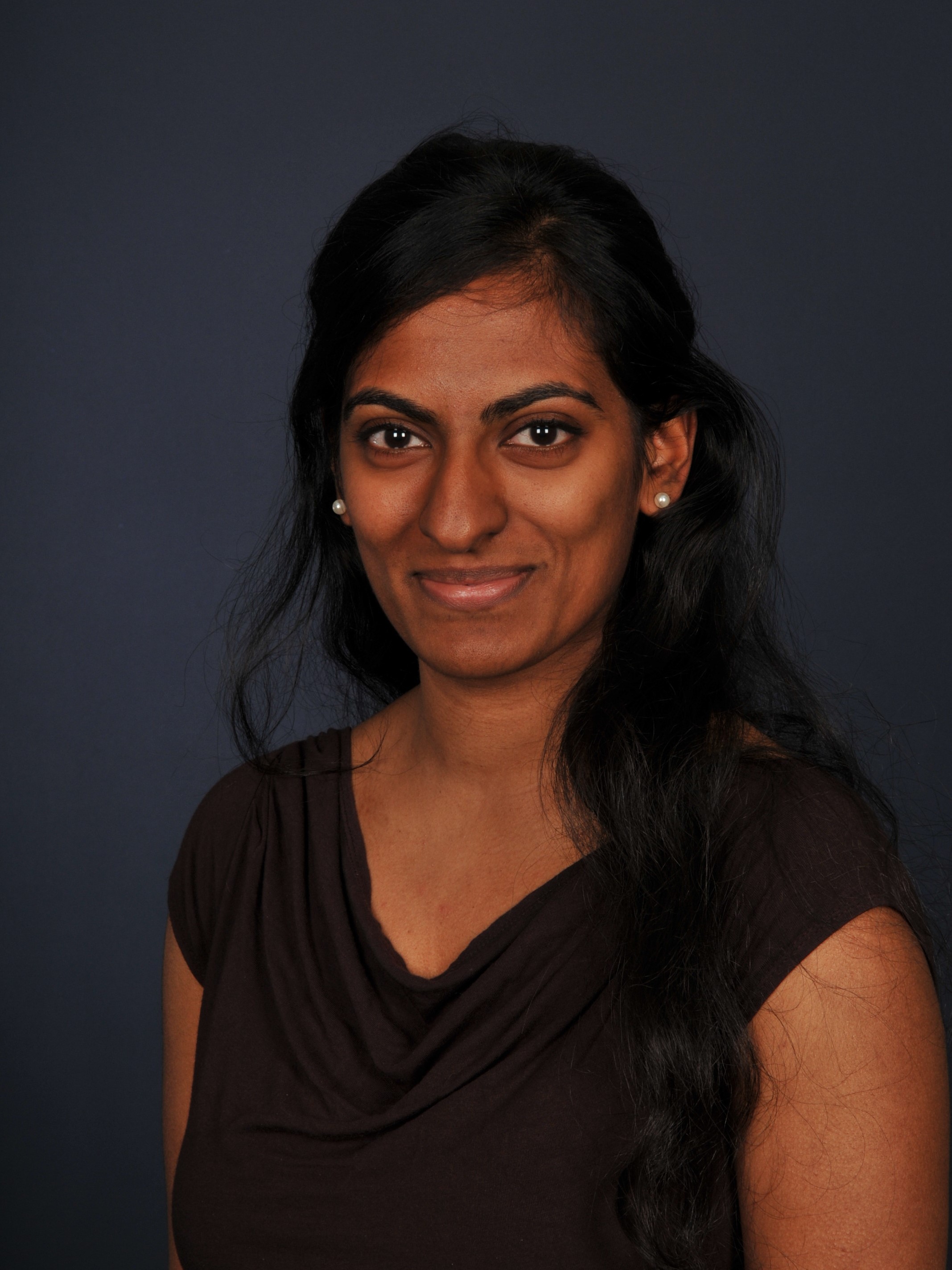}}]{Thulasi Mylvaganam} was born in Bergen, Norway, in 1988. She received the M.Eng. degree in Electrical and Electronic Engineering and the Ph.D. degree in Nonlinear Control and Differential Games from Imperial College London, London, U.K., in 2010 and 2014, respectively. From 2014 to 2016, she was a Research Associate with the Department of Electrical and Electronic Engineering, Imperial College London. From 2016 to 2017, she was a Research Fellow with the Department of Aeronautics, Imperial College London, UK, where she is currently Associate Professor. Her research interests include nonlinear control, optimal control, game theory, distributed control and data-driven control. She is Associate Editor of the IEEE Control Systems Letters, of the European Journal of Control, of the IEEE CSS Conference Editorial Board and of the EUCA Conference Editorial Board. She is also Vice Chair of Education for the IFAC Technical Committee 2.4 (Optimal Control) and Member of the UKACC Executive Committee.
\end{IEEEbiography}

\begin{IEEEbiography}[{\includegraphics[width=1in,height=1.25in,clip,keepaspectratio]{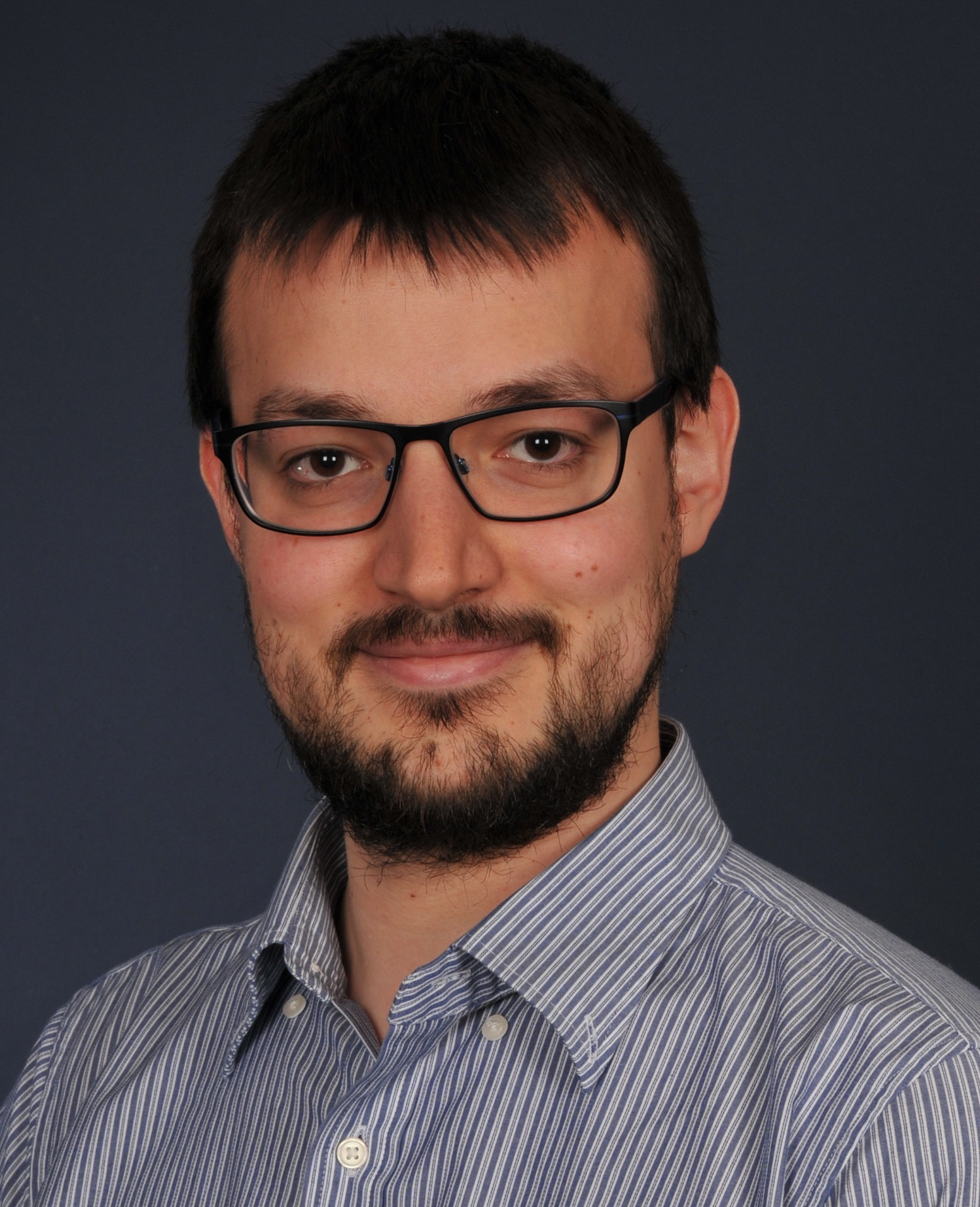}}]{Giordano Scarciotti} (Senior Member, IEEE) received his B.Sc. and M.Sc. degrees in Automation Engineering from the University of Rome ``Tor Vergata'', Italy, in 2010 and 2012, respectively. In 2012 he joined the Control and Power Group, Imperial College London, UK, where he obtained a Ph.D. degree in 2016. 
He also received an M.Sc. in Applied Mathematics from Imperial in 2020. He is currently an Associate Professor at Imperial. 
He was a visiting scholar at New York University in 2015, at University of California Santa Barbara in 2016, and a Visiting Fellow of Shanghai University in 2021-2022. He is the recipient of an Imperial College Junior Research Fellowship (2016), of the IET Control \& Automation PhD Award (2016), the Eryl Cadwaladr Davies Prize (2017), an ItalyMadeMe award (2017) and the IEEE Transactions on Control Systems Technology Outstanding Paper Award (2023). He is a member of the EUCA Conference Editorial Board, of the IFAC and IEEE CSS Technical Committees on Nonlinear Control Systems and has served in the International Programme Committees of multiple conferences. He is Associate Editor of Automatica. He was the National Organising Committee Chair for the EUCA European Control Conference (ECC) 2022, and of the 7th IFAC Conference on Analysis and Control of Nonlinear Dynamics and Chaos 2024, and the Invited Session Chair and Editor for the IFAC Symposium on Nonlinear Control Systems 2022 and 2025, respectively. He is the General Co-Chair of ECC 2029. 
\end{IEEEbiography}

\end{document}